\documentclass{article}

\usepackage{alltt}
\usepackage{color}
\usepackage{amsmath,amssymb}
\usepackage{graphicx}          

\definecolor{green2}{rgb}{0,0.62,0.17}
\definecolor{orange}{rgb}{1,0.5,0.0}
\definecolor{b}{rgb}{0,0,1}
\definecolor{r}{rgb}{1,0,0}

\newtheorem{definition}{Definition}
\newtheorem{lemma}{Lemma}
\newtheorem{corollary}{Corollary}
\newtheorem{theorem}{Theorem}
\newtheorem{remark}{Remark}
 \newcommand{\R}{\mathbb{R}}

  \newcommand{\zero}{\mathbf{0}}
 
  \newcommand{\dn}{\mathbf{d}}
\DeclareMathOperator{\rank}{\mathrm{rank}}

\begin{document}

	\title{Fixed-time Stabilization with a Prescribed Constant Settling Time by Static Feedback for Delay-Free and Input Delay Systems}

\author{Andrey Polyakov\footnote{Inria, University of Lille, FR-59000, France (andrey.polyakov@inria.fr)} \;and Miroslav Krstic\footnote{
University of California, San Diego, USA}}
\date{}








\maketitle

\section{Introduction}
The problem of regulation of a system to a desired set-point in a finite time can be solved using, for example, the methods of finite-time stabilization (see, e.g., \cite{EfimovPolyakov2021:Book} and references therein). Algorithms of finite-time regulation/stabilization for linear systems  are well-known since 1950 (see, for example,  \cite{Feldbaum1953:AiT}, \cite{LaSalle1958:PNAS}, \cite{Fuller1960:IFAC}, \cite{Korobov1979:DAN}, \cite{Haimo1986:SIAM_JCO}, \cite{BhatBernstein2005:MCSS}, \cite{Andrieu_etal2008:SIAM_JCO}).
The settling time to a set-point may be uniformly bounded for all initial conditions (see, e.g., \cite{Majda1975:IUMJ}, \cite{Balakrishnan2005:AMC}, \cite{Andrieu_etal2008:SIAM_JCO}). 
In \cite{Polyakov2012:TAC}, such a property of finite-time stable systems was called fixed-time stability. 

Both time-independent (static) feedbacks (see, e.g., \cite{Polyakov2012:TAC}, \cite{Zimenko_etal2020:TAC}) and time-dependent regulators (see, e.g., \cite{Song_etal2017:Aut}, \cite{Orlov2022:Aut}) are developed for fixed-time stabilization/regulation of linear plants. For controllable systems,  the settling time can be tuned arbitrary small. The latter immediately follows from the definition of controllability. For a control system topologically equivalent to the integrator chain, very simple schemes for tuning of the settling time are given, for example, in  \cite{Song_etal2017:Aut} and \cite{PolyakovKrstic2023:TAC}.
The time-dependent feedback \cite{Song_etal2017:Aut} is designed such that the closed-loop system  converges to the origin  \textit{exactly} at a desired (prescribed) time $T>0$ independently of the initial condition (away from the origin). This property is, obviously, more strong than simply a fixed-time stabilization in a prescribed time $T>0$. In the latter case, the system reaches the desired set-point no later than the time instant $t=T$. An assignment of the exact constant settling time for all initial condition may be useful for certain control problems \cite{Song_etal2017:Aut}, \cite{Shinar_etal2014:JFI},
where the reaching of the set-point before or after the given time is an undesirable event. 

The time-dependent (prescribed-time) regulators  are designed for various finite dimensional \cite{Song_etal2017:Aut}, \cite{Holloway2018:PhD}, \cite{Abel_etal2022:ACC} and infinite dimensional \cite{Espitia_etal2019:Aut}, \cite{Steeves_etal2020:EJC} systems. However, in general, their structure  is the same in all cases, namely, the control law has the form of a linear feedback with a {\it time-dependent gain tending to infinity} as the time tends to the prescribed time $T>0$. This definitely impacts the robustness properties of the closed-loop system despite that the closed-loop system satisfies the ISS\footnote{ISS=Input-to-State Stability\cite{Sontag1989:TAC}}-like estimates \cite{Song_etal2017:Aut} on the prescribed interval of time $[0,T)$.
In the delay-free case, the mentioned time-dependent controller rejects matched additive disturbances of unknown magnitude \cite{Song_etal2017:Aut}, but it is very sensitive with respect to measurement noises \cite{Orlov2022:Aut}. The main reason of such sensitivity is the time dependence  of the feedback gain which, independently of the stabilization error and the magnitude of the measurement noises, infinitely amplifies the impact of the noise to the closed-loop dynamics as time tends to the desired settling time $T$. To improve the robustness, a switching rule between time-dependent prescribed-time regulator and a static finite-time (sliding mode) stabilizer has been suggested in \cite{Orlov2022:Aut}. 

  The fixed-time stabilizer presented in \cite{PolyakovKrstic2023:TAC}  is a static (\textit{time-independent}) nonlinear feedback, which can be interpreted as a linear control with a \textit{state dependent feedback gain}. This gain tends to infinity  as the norm of the stabilization error tends to zero. Such a control system admits a simple scheme for tuning of the stabilization time, but it does not allow us to assign a prescribed constant settling time of the stabilization. The controller is known to be robust (in the ISS sense) with respect to a rather large class of perturbations \cite{Andrieu_etal2008:SIAM_JCO}, \cite{Zimenko_etal2020:TAC}, \cite{Polyakov_etal2023:Aut} including measurement noises. 
  Moreover, comparing with the time-dependent stabilizer, it is expected to be less sensitive with respect to measurement noises since the feedback  gain (amplifying the noise) 
  does not tend to infinity in this case (due to non-zero stabilization error). 

In this paper we design a global static feedback, which stabilizes the linear MIMO system  such that the settling time of the closed-loop system to zero  equals exactly to a prescribed time $T$ for all non-zero initial conditions. To the best of authors' knowledge, the static (time-independent) feedbacks  solving the mentioned problem have never been designed before, probably due to the following reason.  
The finite-time stability with a constant settling time is impossible for continuous autonomous ODE (Ordinary Differential Equation), since the settling-time function of any finite-time stable ODE is \textit{strictly} decreasing along non-zero trajectories of the system \cite{BhatBernstein2000:SIAM_JCO}. So, it cannot be a constant for all non-zero  initial conditions. Therefore, the considered problem of stabilization is infeasible by a conventional static nonlinear feedback. Inspired by \cite{Polyakov_etal2023:Aut}, to overcome this fundamental obstacle, we define the \textit{gain of a static (time-invariant) nonlinear feedback as a function of the initial state}. Formally, in this case, the closed-loop system becomes a Functional  Differential Equation (FDE), since its right-hand side depends on both current  and  previous (more precisely, initial) values of the state vector.  However, this is a very particular class of FDE, since for any fixed initial condition, the FDE becomes an ODE and can be analyzed in the conventional way. Our design is essentially-based on an extension of  homogeneity concept to such a class of FDEs.

Homogeneity is a dilation symmetry widely utilized \cite{Zubov1958:IVM},\cite{Rosier1992:SCL}, \cite{Hong2001:Aut}, \cite{BhatBernstein2005:MCSS}, \cite{Orlov2005:SIAM_JCO}, \cite{Levant2005:Aut}, \cite{Perruquetti_etal2008:TAC}, \cite{Andrieu_etal2008:SIAM_JCO}, \cite{Polyakov2020:Book} for finite-time stabilization and stability analysis, since  any asymptotically stable homogeneous system of negative degree is finite-time stable. This paper extends the homogeneity-based analysis to a particular class of FDE, which can be treated as autonomous ODEs with right-hand sides depended on the initial state.
In this case, the vector field (the right-hand side of the FDE) may be homogeneous with respect a scaling of both the actual and initial state vectors. We show that, under certain conditions, the asymptotically stable homogeneous FDE is  finite-time stable with a constant settling time function. This novel result extends the existing knowledge about convergence rates of homogeneous systems.

 The novel fixed-time controllers are designed for both delay-free and input delay LTI systems.  In \cite{Zekraoui_etal2023:Aut}, the fixed-time stabilizer (with non-constant settling time) for the integrator chain has been designed using transport PDE (Partial Differential Equation) as a model of the input delay and the back-stepping transformation \cite{Krstic2009:Book}, \cite{KarafyllisKrstic2017:Book}. However, a similar PDE-based analysis seems impossible for our fixed-time stabilizer due to its discontinuity. 
In the delay-free case, the analysis of the closed-loop dynamics is based on the Filippov's theory of discontinuous differential equations\cite{Filippov1988:Book}. Its analog for PDEs with discontinuous controllers is not yet well-developed, despite of some interesting recent contributions to this field \cite{Orlov2020:Book}. We extend the results obtained in the delay-free case to the input delay LTI system by means of the Artstein's transformation \cite{Artstein1982:TAC}, which allows both the stability and the robustness analysis of the closed-loop system to be realized easily.

The paper is organized as follows. First, the problem statement is presented. Next, the preliminaries about a particular class of homogeneous FDEs is presented. After that, a fixed-time  stabilizer with a prescribed constant settling time is designed for LTI system. Finally, the numerical simulation examples and conclusions are presented.

\textit{Notation}.

$\R$ is the field of reals; $\R^n_{\zero}=\R^{n}\backslash\{\zero\}$, where $\zero$ is the zero element of a vector space (e.g., $\zero \in \R^n$ means that $\zero$ is the zero vector); $\|\cdot\|$ is a norm in $\R^n$ (to be specified later);
a matrix norm for $A\in \R^{n\times n}$ is defined as $\|A\|=\sup_{x\neq \zero} \frac{\|Ax\|}{\|x\|}$;  $\lambda_{\min}(P)$ denote a minimal eigenvalue of a symmetric matrix $P=P^{\top}\in \R^{n\times n}$; $P\succ 0$ means that the symmetric matrix $P$ is positive definite; $C^1(\Omega_1,\Omega_2)$ denotes the set of continuously differential functions  $\Omega_1\subset \R^n\mapsto \Omega_2\subset \R^m$; $L^{\infty}(\R,\R^k)$ is the Lebesgue space of measurable uniformly essentially bounded functions $\R\mapsto \R^k$ with the norm defined by the essential supremum;  $W^{1,\infty}(\Gamma,\R^n)=\{\phi\in L^{\infty}(\Gamma,\R^n): \dot \phi \in L^{\infty}(\Gamma,\R^n)\}$ is the Sobolev space of absolutely continuous functions $\Gamma\subset \R \mapsto \R^n$; $\|q\|_{L^{\infty}_{(t_0,t)}}=\mathrm{ess}\sup_{\tau\in (t_0,t)}\|q(\tau)\|$ for $q\in  L^{\infty}(\R,\R^k)$; we write  $\stackrel{a.e.}{=}$ (resp, $\stackrel{a.e.}{\leq}$ or $\stackrel{a.e.}{\in}$) if an identity (resp., inequality or inclusion) holds almost everywhere.


\section{Problem statement}
Let us consider the system
\begin{equation}\label{eq:mainsystem}
	\dot x(t)=Ax(t)+Bu(t-\tau), \quad t>0, \quad x(0)=x_0\in \R^n,
\end{equation}
where $x(t)\in \R^n$ is the state variable, $u(t-\tau)\in \R^m$ is the control signal, $A\in \R^{n\times n}$ and $B\in \R^{n\times m}$ are known matrices, the time shift $\tau\geq 0$ models a delay of a transmission of the input signal to the plant. We restrict the class of admissible control signals $u\in L^{\infty}((-\tau,+\infty),\R^m)$, so the differential equation in \eqref{eq:mainsystem} is assumed to be fulfilled almost everywhere. 
Notice that $u$ has to be defined on the time interval $(-\tau,+\infty)$ to guarantee the well-possedness of the system \eqref{eq:mainsystem}. The whole state vector $x(t)$ is assumed to be available (measured or estimated) for the control peroposes.
The pair $\{A,B\}$ is assumed to be controllable.

First, we  study the delay-free case ($\tau=0$).  For a given constant $T>0$, we need to design a nonlinear feedback
\begin{equation}\label{eq:u_delay_free}
	u(t)=\tilde K(x(t),x_0)x(t), \quad \tilde K\in C^1(\R^n_{\zero} \times \R^n, \R^{m\times n})
\end{equation}
 such that the closed-loop system \eqref{eq:mainsystem}, \eqref{eq:u_delay_free}  is {\it fixed-time stable with a constant settling time $T>0$}. The latter means that the system is Lyapunov  stable and $x(t)=\zero$, $\forall t\geq T$, $\forall x_0\in \R^n$, but $x(t)\neq \zero$ for $t\in [0,T)$ if $x_0\neq \zero$. Therefore, any trajectory of the system initiated away from the origin will reach the origin exactly at the time instant $t=T$.  

The second goal of the paper is to study the robustness  (in the sense of Input-to-State Stability \cite{Sontag1989:TAC})  of the closed-loop system with respect to additive disturbances and measurement noises.

Finally, the third goal is to study  the above problems for the input delay system \eqref{eq:mainsystem} with $\tau>0$. In this case,
the control value generated at time $t$ affects the system in the future time instant $t+\tau$. Since the plant model \eqref{eq:mainsystem} is valid only for $t>0$, the control signal $u(t)$ can be generated  based on the state measurements (similarly to \eqref{eq:u_delay_free}) only for $t>0$, but for $t\in (-\tau,0)$ it has to be initialized as follows:
\begin{equation}\label{eq:initial_condition}
  u(\theta)=\phi(\theta)\quad \text{ for }\quad  \theta\in (-\tau,0), \quad \phi\in L^{\infty}((-\tau,0),\R^m),
\end{equation}
where, dependently of the control application, the initial control signal $\phi$ can be assumed to be uncertain or assigned as needed for reaching the control goal.
Due to the input delay, we restrict the desired settling time to $T+\tau>0$.
Below we show that the problem of the fixed-time stablization with a constant settling time  $T+\tau$ is feasible only if the control signal on $(-\tau,0)$ is initialized by the zero ($\phi=\zero$). Notice that, for an arbitrary selected or unknown initial function $\phi\neq \zero$, the fixed-time stabilization with the prescribed time $T+\tau$ remains possible. However, the constant value  of the settling time cannot be guaranteed anymore.  The settling time can be  just bounded by the prescribed constant $T+\tau$  in  this case.


\section{Preliminaries}\label{sec:pre}
\subsection{Stability notions}
Let us consider the system
\begin{equation}\label{eq:fde}
	\dot x(t)=f(x(t),x(t_0)), \quad t>t_0, \quad x(t_0)=x_0\in \R^n,
\end{equation}
where $x(t)$ is the system state and $f: \R^n\times \R^n\mapsto \R^n$ is  locally bounded.
On the one hand, the system \eqref{eq:fde} is well-posed, since it can be rewritten in the form of the conventional ODE
\begin{equation}
\left\{
\begin{array}{l}
	\dot x=f(x,r),\\
	\dot r=\zero,
\end{array}
\right. \quad t>t_0,  \quad x(t_0)=r(t_0)=x_0.
\end{equation}
For simplicity, we assume that  $f\in C(\R^n_{\zero} \times \R^n, \R^n)$.  However, the results presented in this section are also valid for discontinuous equations and inclusions studied in  \cite{Filippov1988:Book}. If $f\in C(\R^n_{\zero} \times \R^n, \R^n)$ so the latter system (as well as the system \eqref{eq:fde}) has  classical (possible non-unique) solutions $t\mapsto (x(t),r(t))$ on $\R^n_{\zero}\times \R^n$ and Filippov solutions on $\R^n_{\zero}\times \R^n$ provided that $f$ is locally bounded.

 On the other, the differential equation \eqref{eq:fde} is not a dynamical system in the sense that {\it its solutions do not satisfy the so-called semi-group property \cite{MironchenkoPrieur2020:SIAM}}. Namely,  if $x(t,t_0,x_0)$ with $t\geq t_0$ denotes a solution of the system \eqref{eq:fde} then, in the general case, $x(t,s,x(s,t_0,x_0))\neq x(t+s,t_0,x_0)$, where $s>t_0$.
So, the classical results of the stability theory (such as the Lyapunov function method) cannot be directly applied to the system \eqref{eq:fde}.
However, the stability notions can be introduced in the conventional way.
Since, in this paper, we deal only with a global uniform stability, then, for shortness, we omit the words "global uniform"  when we discuss stability issues.
\begin{definition}\cite{EfimovPolyakov2021:Book}\label{def:stab}
	The system \eqref{eq:fde} is said to be 
	\begin{itemize}
		\item Lyapunov stable if there exists $\varepsilon\in \mathcal{K}$ such that 
		\begin{equation}
		\|x(t)\|\leq \varepsilon(\|x_0\|), \quad \forall t\geq t_0, \quad \forall x_0\in \R^n;
		\end{equation}
		\item asymptotically stable if there exists $\beta\in \mathcal{KL}$
		\begin{equation}
		\|x(t)\|\leq \beta(\|x_0\|,t-t_0), \quad \forall x_0\in \R^n, \quad \forall t\geq t_0.
		\end{equation}
		\item finite-time stable if it is   Lyapunov stable and the exists a  locally bounded function $\tilde T: \R^n\mapsto \R_+$  such that
		for any $x_0\in \R^n\backslash\{\zero\}$  it holds
		\begin{equation}
		x(t)=\zero, \quad \forall t\geq t_0+\tilde T(x_0),\quad  
		\end{equation}
		for any solution $x(t)$ of \eqref{eq:fde} with  $x(t_0)=x_0$, but
		 $x(t)\neq \zero $ for all $t\in [0,\tilde T(x_0))$, at least, for one solution $x(t)$ of \eqref{eq:fde} with  $x(t_0)=x_0$;
		\item fixed-time stable if it is finite-time stable and there exists $T_{\max}>0$ such that
		\begin{equation}
		\exists T_{\max}>0\quad : \quad \tilde T(x_0)\leq T_{\max}, \quad \forall x_0\in \R^n;
		\end{equation}
	\end{itemize}
\end{definition}
The function $\tilde T$ from the latter definition is known as the \textit{settling time function} \cite{BhatBernstein2000:SIAM_JCO} and its value $\tilde T(x_0)$ is referred to as the \textit{settling time} for the given initial state $x_0$. 
In this paper, we study \textit{finite-time stable systems  with constant settling-time functions}, i.e., $\tilde T(x_0)=\mathrm{const}$ for all $x_0\neq \zero$.

Let us consider the system
\begin{equation}\label{eq:fde_q}
	\dot x(t)\stackrel{a.e.}{=}\tilde f(x(t),x_0,q(t)), \quad t>t_0, \quad x(t_0)=x_0\in \R^n,
\end{equation}
where $x(t)$ is the system state, $q\in L^{\infty}(\R,\R^k)$ and $\tilde f$ is a locally bounded measurable  function such that the system \eqref{eq:fde_q} has a strong (Carath\'eodory or Filippov) solution for any  $x_0\in \R^n$, any $t_0\in \R$ and any $q\in L^{\infty}(\R,\R^k)$. 

\begin{definition}\cite{Sontag1989:TAC}
	A system  is said to be Input-To-State Stable (ISS) if 
	there exist  $\beta\in \mathcal{KL}$ and $\gamma\in \mathcal{K}$ 
	such that 
	\begin{equation}\label{eq:ISS}
		\|x(t)\|\!\leq\! \beta(\|x_0\|,t-t_0)+\gamma(\|q\|_{L^{\infty}_{(t_0,t)}}), \, \forall x_0\!\in\! \R^n, \, \forall t\!\geq\!t_0, \, \forall q\!\in\! L^{\infty}(\R,\R^k), \, \forall t_0\!\in\! \R,
	\end{equation}
where $x$ is a state of the system at the time $t\geq t_0$, $x_0$ is the initial state and $q$ is an exogenous  input/perturbation. 
\end{definition}
In control theory, the ISS is frequently interpreted as a robustness of the system with respect to a perturbation, which is modelled by an exogenous input $q$ in the right-hand side. Since the right-hand side of \eqref{eq:fde_q} depends on the initial condition then the behavior of the perturbed system on the infinite horizon could also depend on the initial state $x_0$.   Therefore, \textit{in the case of the system \eqref{eq:fde_q}, the function $\gamma$ in the ISS definition may depend on $x_0$ as well}.
	
\subsection{Homogeneous Systems}
\subsection{Linear dilations}
Let us recall  that \textit{a family of  operators} $\dn(s):\R^n\mapsto \R^n$ with $s\in \R$ is  a \textit{group} if
$\dn(0)x\!=\!x$, $\dn(s) \dn(t) x\!=\!\dn(s+t)x$, $\forall x\!\in\!\R^n, \forall s,t\!\in\!\R$.
A \textit{group} $\dn$ is \\
a) \textit{continuous} if the mapping $s\mapsto \dn(s)x$ is continuous,  $\forall x\!\in\! \R^n$;\\
b) \textit{linear} if $\dn(s)$ is a linear mapping (i.e., $\dn(s)\in \R^{n\times n}$), $\forall s\in \R$;\\
c)  a \textit{dilation} in $\R^n$ if $\liminf\limits_{s\to +\infty}\|\dn(s)x\|=+\infty$ and $\limsup\limits_{s\to -\infty}\|\dn(s)x\|=0$,  $\forall x\neq \zero$.

Any linear continuous group in $\R^n$ admits the representation  \cite{Pazy1983:Book}
\begin{equation}\label{eq:Gd}
	\dn(s)=e^{sG_{\dn}}=\sum_{j=1}^{\infty}\tfrac{s^jG_{\dn}^j}{j!}, \quad s\in \R,
\end{equation}
where $G_{\dn}\in \R^{n\times n}$ is a generator of $\dn$. A continuous linear group \eqref{eq:Gd} is a dilation in $\R^n$ if and only if $G_{\dn}$ is an anti-Hurwitz matrix \cite{Polyakov2020:Book}. In this paper we deal only with continuous linear dilations. 
A \textit{dilation} $\dn$ in $\R^n$ is\\
i) \textit{monotone} if the function $s\mapsto \|\dn(s)x\|$ is strictly increasing,  $\forall x\neq \zero$;\\
ii) 	\textit{strictly monotone} if  $\exists \beta\!>\!0$ such that $\|\dn(s)x\|\!\leq\! e^{\beta s}\|x\|$, $\forall s\!\leq\! 0$, $\forall x\in \R^n$.\\
The following result is the straightforward consequence of the existence of the quadratic Lyapunov function  for asymptotically stable LTI systems.
\begin{corollary}\label{cor:monotonicity}
	A linear continuous dilation in $\R^n$ is strictly monotone with respect to the weighted Euclidean norm $\|x\|=\sqrt{x^{\top} Px}$ with $0\prec P\in \R^{n\times n}$ if and only if 
	\begin{equation}\label{eq:mon_d_P}
		PG_{\dn}+G_{\dn}^{\top}P\succ 0, \quad P\succ 0.
	\end{equation}
\end{corollary}

\subsection{Canonical homogeneous norm}
Any linear  continuous and monotone  dilation in a normed vector space introduces also an alternative norm topology defined by the so-called canonical homogeneous norm \cite{Polyakov2020:Book}.
\begin{definition}[{\it Canonical homogeneous norm}]
	\label{def:hom_norm_Rn}
	Let a linear dilation $\dn$ in $\R^n$ be  continuous and monotone with respect to a norm $\|\cdot\|$.
	A function $\|\cdot\|_{\dn} : \R^n \mapsto [0,+\infty)$ defined as  follows: $\|\zero\|_{\dn}=0$ and 
	\begin{equation}\label{eq:hom_norm_Rn}
		\|x\|_{\dn}=e^{s_x}, \;  \text{where} \; s_x\in \R: \|\dn(-s_x)x\|=1, \quad x\neq \zero
		\vspace{-1mm}
	\end{equation}
	is said to be a canonical $\dn$-homogeneous norm \index{canonical homogeneous norm} in  $\R^n$ 
\end{definition}	 
Notice that, by construction, $\|x\|_{\dn}=1\; \Leftrightarrow\;  \|x\|=1$. Due to the monotonicity of the dilation, $\|x\|_{\dn}< 1 \; \Leftrightarrow\;  \|x\|<1$ and $\|x\|_{\dn}>1 \; \Leftrightarrow\;  \|x\|>1$.

For standard dilation $\dn_1(s)=e^{s}I_n$ we, obviously, have $\|x\|_{\dn_1}=\|x\|$. In other cases, $\|x\|_{\dn}$ with $x\neq \zero$ is implicitly defined by a nonlinear algebraic equation, which always have a unique solution due to monotonicity of the dilation. In some particular cases \cite{PolyakovKrstic2023:TAC}, this implicit equation has explicit solution even for non-standard dilations. 
\begin{lemma}\cite{Polyakov2020:Book}\label{lem:hom_norm}
	If  a linear continuous dilation $\dn$ in $\R^n$  is  monotone with respect to a norm $\|\cdot\|$
	then
	\begin{itemize}
		\item[1)]   $\|\cdot\|_{\dn} : \R^n\mapsto \R_+$ is single-valued and continuous on $\R^n$;
		\item[2)]
		there exist $\sigma_1,\sigma_2 
		\in \mathcal{K}_{\infty}$
		such that 
		\begin{equation}\label{eq:rel_norm_and_hom_norm_Rn}
			\sigma_1(\|x\|_{\dn})\leq \|x\|\leq \sigma_2(\|x\|_{\dn}), \quad \quad 
			\forall x\in \R^n;
		\end{equation}
		\item[3)]   $\|\cdot\|$ is locally Lipschitz continuous on $\R^{n}\backslash\{\zero\}$ provided that  the linear dilation $\dn$ is strictly monotone
		\item[4)] $\|\cdot\|_{\dn}$ is continuously differentiable on $\R^n\backslash\{\zero\}$ provided that $\|\cdot\|$ is continuously differentiable on $\R^n\backslash\{\zero\}$ and $\dn$ is strictly monotone.
	\end{itemize}
\end{lemma}
For the $\dn$-homogeneous norm $\|x\|_{\dn}$ induced by the weighted Euclidean norm $\|x\|=\sqrt{x^{\top}Px}$ we have \cite{Polyakov2020:Book}
\begin{equation}\label{eq:hom_norm_der}
	\tfrac{\partial \|x\|_{\dn}}{\partial x}=\|x\|_{\dn}\tfrac{x^{\top}\dn^{\top}(-\ln \|x\|_{\dn})P\dn(-\ln \|x\|_{\dn})}{x^{\top}\dn^{\top}(-\ln \|x\|_{\dn})PG_{\dn}\dn(-\ln \|x\|_{\dn})x}.
\end{equation}

\subsubsection{Homogeneous Functions and Vector fields}
Below we study systems that are symmetric with respect to a linear dilation $\dn$. The  dilation symmetry introduced by the following definition is known as a generalized  homogeneity \cite{Zubov1958:IVM}, \cite{Kawski1991:ACDS}, \cite{Rosier1992:SCL}, \cite{BhatBernstein2005:MCSS}, \cite{Polyakov2020:Book}.
\begin{definition}\cite{Kawski1991:ACDS}\label{def:hom_vf}
	A vector field $f:\R^n \mapsto \R^n$ is  $\dn$-homogeneous of 
	degree $\mu\in \R$  if 
	\begin{equation}\label{eq:homogeneous_operator_Rn}
		g(\dn(s)x)=e^{\mu s}\dn(s)g(x), \quad  \quad  \forall s\in\R, \quad \forall x\in \R^n.
	\end{equation}
\end{definition}

The  homogeneity of a mapping is inherited by other mathematical objects induced by this mapping.
In particular, solutions of $\dn$-homogeneous system\footnote{A system is homogeneous  if its is governed by a $\dn$-homogeneous vector field}
\begin{equation}\label{eq:hom_system}
	\dot x=g(x),\quad  t>0, \quad x(0)=x_0\in \R^n
\end{equation}
are symmetric with respect to the dilation $\dn$ in the following sense \cite{Zubov1958:IVM}, \cite{Kawski1991:ACDS}, \cite{BhatBernstein2005:MCSS}
\begin{equation}
	x(t,\dn(s)x_0)=\dn(s)x(e^{\mu s}t,x_0),
\end{equation}  
where $x(\cdot,z)$ denotes a solution of  \eqref{eq:hom_system}  with $x(0)=z\in \R^n$ and 
$\mu\in \R$ is the homogeneity degree of $g$. The mentioned symmetry  
of solutions implies many useful properties of homogeneous system such as equivalence of local and global results. For example, local asymptotic (Lyapunov or finite-time stability) is equivalent to global asymptotic (resp., Lyapunov or finite-time) stability.

\subsubsection{Homogeneous FDE}
It is well known \cite{Zubov1964:Book}, \cite{Nakamura_etal2002:SICE}, \cite{BhatBernstein2005:MCSS} that an asymptotically stable system \eqref{eq:hom_system} is finite-time stable provided that $g$ is $\dn$-homogeneous of negative degree. 
The following theorem shows that a homogeneous FDE \eqref{eq:fde} may be fixed-time stable with a constant settling-time function.
\begin{theorem}\label{thm:FxT_PT_hom}
	Let $f\in C(\R^n_{\zero} \times \R^n, \R^n)$ be locally bounded and the system 
	\begin{equation}\label{eq:fde_r}
		\dot x=f(x,r), \quad t>0,\quad x(0)=x_0, \quad r\neq \zero
	\end{equation}
	be asymptotically stable for any $r\in \R^n_{\zero}$.
	Let $\dn_1, \dn_2$ be  linear dilations in $\R^n$ such that 
	\begin{itemize}
		\item for any $r\in \R^n$ the vector field 
  \begin{equation}x\mapsto f(x,r)
  \end{equation} is $\dn_1$-homogeneous of negative degree $\mu<0$;
		\item the vector field 
		\begin{equation}
		    \left(\begin{array}{c} x\\r\end{array}\right)\mapsto \left(\begin{array}{c}
			f(x,r)\\
			\zero \end{array}\right), \quad x,r\in \R^n
		\end{equation}
		is $\tilde \dn$-homogeneous of degree $0$, where
		\begin{equation}
		\tilde \dn(s)=\left(\begin{array}{cc} \dn_2(s) & \zero\\ \zero & \dn_2(s)\end{array}\right), \quad s\in \R.
		\end{equation}
	\end{itemize}
 If the system \eqref{eq:fde} is Lyapunov stable then it is 
	\begin{itemize}
     	\item finite-time stable with a discontinuous (at least at $\zero$) settling-time function $\tilde T:\R^n\mapsto [0,+\infty)$;    
		\item  fixed-time stable provided that $\tilde T$ is bounded on some compact set $S\subset \R^n_{\zero}$ such that $\bigcup\limits_{s\in \R} \dn_2(s) S=\R^{n}_{\zero}$, moreover, the settling time is a constant for all $x_0\neq \zero$ if and only if  $\tilde T$ is constant on $S$.
	\end{itemize}
\end{theorem} 
\begin{proof}
	On the one hand, since for any fixed $r\neq \zero$ the system \eqref{eq:fde_r} is asymptotically stable then 
	$\dn$-homogeneity with negative degree $\mu<0$ implies its finite-time stability (see \cite{BhatBernstein2005:MCSS}, \cite{Orlov2005:SIAM_JCO}), i.e.,
there exists a settling-time function $T_r: \R^n\mapsto \R_+$. This means that $T_{x_0}(x_0)<+\infty$ for any $x_0\neq \zero$. Hence, the Lyapunov stability of the system \eqref{eq:fde} implies its  finite-time stability with  the settling-time function 
$\tilde T:\R^{n}\mapsto [0,+\infty)$ defined as follows
	\begin{equation}
	\tilde T(x_0)=T_{x_0}(x_0) \quad\text{ for } \quad  x_0\neq \zero
	\end{equation}
	and $\tilde T(\zero)=\zero$. 
	
	

	On the other hand, $\tilde \dn$-homogeneity of the system \eqref{eq:fde_r} implies the following dilation symmetry of solutions
	\begin{equation}
	x_{\dn_2(s)r}(t,\dn_2(s)x_0)=\dn_2(s)x_{r}(t,x_0), \quad t\geq 0
	\end{equation}
	where $x_{\tilde r}(t,z)$ denotes the solution of the system \eqref{eq:fde_r}   with the initial condition 
	$x(0)=z$ and the vector of parameters (in the right-hand side)  $\tilde r\in \R^n$.

	The latter implies that the simultaneous scaling of $x_0$ and $r$ by $\dn_2(s)$ in  \eqref{eq:fde_r} does not change the settling time of the corresponding solution:
	\begin{equation}
	\tilde  T(\dn_2(s)x_0)=T_{\dn_2(s)x_0}(\dn_2(s)x_0)=T_{x_0}(x_0)=\tilde T(x_0)
	\end{equation}
	Therefore, the settling time function of the system \eqref{eq:fde} has a constant value along any homogeneous curve $\Gamma_{\dn_2}(x_0)=\{x\in \R^n: x=\dn_2(s)x_0, s\in \R\}$.
		In this case, $\tilde T$ is always discontinuous at zero
		since $\tilde T(\zero)=0$ and $\tilde T(x_0)\neq 0$ for $x_0\neq 0$. 
	However, taking into account $\bigcup\limits_{s\in \R} \dn_2(s)S=\R^{n}_{\zero}$, a boundedness of $\tilde T$ on the compact $S$ implies the uniform boundedness of $\tilde T$ on $\R^n$, i.e., the system is fixed-time stable.
	Moreover, if $\tilde T(x_0)=$const for all $x_0\in S$ then, due to homogeneity,  $\tilde T(x_0)=$const for all $x_0\in \R^n_{\zero}$.
\end{proof}

Below we design a feedback law $u$ for the system \eqref{eq:mainsystem} with $\tau=0$ such that the closed-loop system satisfies the latter theorem for $S=\{x\in\R^n : \|x\|=1\}$ and $\dn_2(s)=e^sI_n$.

\section{Prescribed-time Stabilization  by Static Homogeneous Feedback}\label{sec:m}
\subsection{Initial State Dependent Homogeneous Feedback}
Inspired by \cite{PolyakovKrstic2023:TAC}, let us consider the system \eqref{eq:mainsystem} and  define the homogeneous feedback as follows
\begin{equation}\label{eq:PT_hom_control}
	u=K_0x+K\dn(-\ln T)\dn\left(-\ln \left\|\frac{x}{\|x_0\|}\right\|_{\dn}\right)x \quad \text{ for } \quad x_0\neq \zero,	
\end{equation}
where $K_0,K\in \R^{m\times n}$ are feedback gains to be defined, $\dn$ is a dilation $\R^n$, $T>0$ is a prescribed settling time of the system, $\|\cdot\|_{\dn}$ is a canonical homogeneous norm induced by a norm $\|\cdot\|$ in $\R^n$ to be defined below. For $x_0=\zero$ we assign $u=K_0x$.

The key difference between the feedback \eqref{eq:PT_hom_control} and the homogeneous controller studied in   \cite{Zimenko_etal2020:TAC} is the  dependence of the feedback gain of the initial state $x_0$ and the prescribed settling time $T>0$. As in   \cite{Zimenko_etal2020:TAC},  the linear term $K_0x$ is  selected such that the linear vector field $x\mapsto (A+BK_0)x$ is $\dn$-homogeneous of negative degree  $\mu=-1$. Moreover, $\dn$ is constructed such that, for any fixed $x_0$, the right-hand side of the closed-loop system \eqref{eq:mainsystem}, \eqref{eq:PT_hom_control} is $\dn$-homogeneous of degree $\mu=-1$. Moreover, it is standard homogeneous of degree $0$ with respect to simultaneous scaling of $x$ and $x_0$ by $\dn_2(s)=e^{s}I_n$. Therefore, in the view of Theorem \ref{thm:FxT_PT_hom}, the asymptotic stability of the closed-loop system implies its fixed-time stability.  Below we prove, by the direct Lyapunov method, that the settling time of the system is equal to the constant $T>0$ for any 
initial state $x(0)=x_0\neq \zero$. Since the feedback \eqref{eq:PT_hom_control} is discontinuous at $x=0$ (for $x_0\neq \zero$) then solutions of the closed-loop system \eqref{eq:mainsystem}, \eqref{eq:PT_hom_control} are defined in the sense of Filippov \cite{Filippov1988:Book}. 
\begin{lemma}[Well-posedness of delay-free control system]\label{lem:wp_delay_free} The feedback law \eqref{eq:PT_hom_control} is locally bounded and  for any $x\neq \zero$ it holds
	\begin{equation}
	u\to K_0x \quad \text{ as } \quad x_0\to \zero.
	\end{equation}
	For any $x_0\neq \zero$, the closed-loop system \eqref{eq:mainsystem}, \eqref{eq:PT_hom_control} has a global-in-time Filippov solution $x:\R_+ \mapsto \R^n$ being a unique classical solution as long as $x(t)\neq \zero$.
 For $x_0=\zero$, the closed-loop system has the unique zero solution.
\end{lemma}
\begin{proof} 
According to Filippov's method \cite{Filippov1988:Book}, to define a solution of the closed-loop system, the discontinuous feedback is regularized at $x=\zero$ as follows
\begin{equation}\label{eq:Filippov_regularization}
u\in 
\|x_0\|K\dn(-\ln T) \mathcal{B} \quad \text{ for } \quad x=\zero,
\end{equation}
where 
$\mathcal{B}=\{x\in \R^n: \|x\|\leq 1\}$ is the unit ball.
The right-hand side of the closed-loop system  \eqref{eq:mainsystem}, \eqref{eq:PT_hom_control}  becomes a differential inclusion with an upper semi-continuous right-hand side, which is single-valued at $x\neq \zero$ and set-valued (compact- and convex-valued) at $x=0$. In this case, the system has a Filippov solution (defined at least locally in time) for any $x_0\neq \zero$ (see, \cite{Filippov1988:Book} for more details).

Since, by definition of the canonical homogeneous norm, we have  $$\|\dn(-\ln \|z\|_{\dn})z\|=1$$ then 
	\begin{equation}
	\|u-K_0x\|_{\R^m}\leq \|K\dn(-\ln T)\|_{m}\cdot \|x_0\| 
	\end{equation}
	where $\|\cdot\|_{\R^m}$ is a norm in $\R^m$, $\|\cdot\|$ is a norm utilized for the definition of $\|\cdot\|_{\dn}$ and $\|K\dn(-\ln T)\|_m=\sup_{\|x\|=1} \|K\dn(-\ln T)x\|_{\R^m}$.
	 Moreover, since the canonical homogeneous norm $\|\cdot\|_{\dn}$ is locally Lipschitz continuous on $\R^{n}_{\zero}$ then the right-hand side of the closed-loop system
	\eqref{eq:mainsystem}, \eqref{eq:PT_hom_control} is locally Lipschitz continuous away from the origin ($x\neq \zero$). So, for any $x_0\neq \zero$, the closed-loop system has a classical solution $x(t), t\geq 0$  defined uniquely as long as $0<\|x(t)\|<+\infty$. 
	Since the right-hand side of the closed-loop system satisfies the estimate 
	\begin{equation}
	    \|Ax+Bu\|\leq \|A+BK_0\|\cdot \|x\|+\|BK\dn(-\ln T)\|\cdot \|x_0\|
	\end{equation}
	with respect to $\|x\|$ then, in the view of Winter's theorem (see, e.g., \cite{GeligLeonovYakubovich2004:Book}), the solution $x(t)$ is defined globally in time (i.e., for all $t\geq0$). For $x_0=\zero$, the (regularized) closed-loop system becomes linear $\dot x=(A+BK_0)x, t>0, x(0)=\zero$, so it has the unique zero solution.
\end{proof}

Notice that, in the view of the latter lemma,  stability of the zero solution should imply the uniqueness of all solutions of the closed-loop system.

\begin{theorem}[\it Homogeneous stabilization with constant settling time]\label{thm:PT_hom_control}\quad \quad\\
	For any controllable pair $\{A, B\}$ one holds
	\begin{itemize}
		\item[\textbf{1)}]  the linear algebraic equation
		\begin{equation}\label{eq:Y_0G_0}
			AG_0-G_0A+BY_0=A, \quad G_0B=\zero
		\end{equation}
		has a solution  $Y_0\!\in\! \R^{m\times n}$, $G_0\!\in\! \mathbb{R}^{n\times n}$  such that 
		\begin{itemize}
			\item the matrix 
			\begin{equation}\label{eq:gen_Gd}
				G_{\mathbf{d}}\!=\!I_n\!+\!\mu G_0
			\end{equation} is anti-Hurwitz for $\mu\!\leq 1/\tilde n$, where $\tilde n\!\in\! \mathbb{N}$ is a minimal number such that $\rank[B,AB,...,A^{\tilde n-1}B]=n$;
			\item the matrix $G_0-I_n$ is invertible and  the matrix 
			\begin{equation}
				A_0=A+BY_0(G_0-I_n)^{-1}
			\end{equation}
			satisfies	the identity 
			\begin{equation}\label{eq:hom_A0}
				A_0G_{\mathbf{d}}=(G_{\mathbf{d}}+\mu I_n)A_0 \quad \text{ and } \quad G_{\mathbf{d}}B=B;
			\end{equation} 		
		\end{itemize}   
		
		\item[\textbf{2)}]  the linear algebraic system  \vspace{-2mm}
		\begin{equation}\label{eq:LMI_Rn}
			A_0X\!+\!XA_0^{\top}\!\!+\!BY\!+\!Y^{\top}\!B^{\top}\!\!+\!G_{\dn}X\!+\!XG_{\dn}^{\top}\!=\!\zero, \quad  G_{\dn}X+XG_{\dn}^{\top}\!\succ\! 0, \quad X\!=\!X^{\top}\!\succ\! 0
		\end{equation}
		always has a solution $X\in \mathbb{R}^{n\times n}$, $Y\in \mathbb{R}^{m\times n}$; \vspace{1mm}
		\item[\textbf{3)}] 
		the closed-loop system \eqref{eq:mainsystem}, \eqref{eq:PT_hom_control} with $\tau=0$, with 
		\begin{equation}
			K_0=Y_0(G_0-I_n)^{-1}, \quad K=YX^{-1}
		\end{equation}
		with 	the dilation $\dn(s)=e^{s G_{\dn}}$ and with the canonical homogeneous norm induced by the norm 
		\begin{equation}
		\|x\|=\sqrt{x^{\top}\dn^{\top}(-\ln T)X^{-1}\dn(-\ln T)x}
		\end{equation}
		is fixed-time stable with  the constant settling time $T>0$;
  \item[\textbf{4)}] all solutions of the closed-loop system \eqref{eq:mainsystem}, \eqref{eq:PT_hom_control} are unique.  
	\end{itemize}
\end{theorem}
\begin{proof}
	The claims 1) and 2) are proven in\cite{Zimenko_etal2020:TAC}.
	For any constant $r>0$, using the formula \eqref{eq:hom_norm_der} we derive
	\begin{equation}
	\tfrac{d}{dt}\|x/r\|_{\dn}=\|x/r\|_{\dn}\tfrac{(x/r)^{\top}\dn^{\top}(-\ln \|x/r\|_{\dn})\dn^{\top}(-\ln T)X^{-1}\dn(-\ln T)\dn(-\ln \|x/r\|_{\dn})(\dot x/r)}{(x/r)^{\top}\dn^{\top}(-\ln \|x/r\|_{\dn})\dn^{\top}(-\ln T)X^{-1}G_{\dn}\dn(-\ln T)\dn(-\ln \|x/r\|_{\dn})(x/r)}.
	\end{equation}
	The identity \eqref{eq:hom_A0} implies that
	\begin{equation}
	A_0\dn(s)=e^{s}\dn(s)A_0 \quad \text{ and } \quad \dn(s)B=e^{s}B, \quad \forall s\in \R
	\end{equation}
	Hence, for the closed-loop system \eqref{eq:mainsystem}, \eqref{eq:PT_hom_control} with $x_0\neq \zero$ we have
	\begin{equation}
	\tfrac{d}{dt}\left\|\tfrac{x}{\|x_0\|}\right\|_{\dn}\!\!=\!\tfrac{1}{T}\tfrac{x^{\top}\!\dn^{\top}\!(-\ln \|x/\|x_0\|\|_{\dn})\dn^{\top}\!(-\ln T)X^{-1}(A_0+BK)\dn(-\ln T)\dn(-\ln \|x/\|x_0\|\|_{\dn})x}{x^{\top}\!\dn^{\top}\!(-\ln \|x/\|x_0\|\|_{\dn})\dn^{\top}\!(-\ln T)X^{-1}G_{\dn}\dn(-\ln T)\dn(-\ln \|x/\|x_0\|\|_{\dn})x}.
	\end{equation}
	Using \eqref{eq:LMI_Rn} we derive
	\begin{equation}
	\tfrac{d}{dt}\left\|\frac{x}{\|x_0\|}\right\|_{\dn}=-\frac{1}{T}
	\end{equation}
	as long as $x(t)\neq -1$. For $t=0$ and $x(0)=x_0\neq \zero$ we have $\|x(0)/\|x_0\|\|_{\dn}=1$.
	Since the derivative of the function $t\mapsto \|x(t)/\|x_0\|\|_{\dn}$ is negative  (for $x(t)\neq \zero$) then 
	$\|x(t)/\|x_0\|\|_{\dn}\leq 1$ for all $t\geq 0$. The latter is equivalent to $\|x(t)\|\leq \|x_0\|$, so the closed-loop system is Lyapunov stable. Moreover,  the system is fied time stable such that $x(t)=\zero$ for $t\geq T$ and $x(t)\neq \zero$ for all $t\in [0, T)$ if $x_0\neq \zero$. Finally, by Lemma \ref{lem:wp_delay_free} any solution of the system unique as long as $x(t)\neq \zero0$, but the proven fixed-time stability guarantee the uniqueness of the solution after the reaching of the origin.  The proof is complete.
\end{proof}


\subsection{Robust stabilization of the delay-free LTI system}

If $x_0$ in \eqref{eq:PT_hom_control} is replaced with a non-zero constant vector, then  the corresponding closed-loop system \eqref{eq:mainsystem} is homogeneous and ISS with respect to measurement noises
(in the view of results \cite{Ryan1995:SCL}, \cite{Andrieu_etal2008:SIAM_JCO}).
So, the obtained static feedback has some ISS properties with respect to measurement noises, but the robustness (namely, the asymptotic gain $\gamma$ in Definition 2) depends essentially  on the initial state $x_0$. Indeed, considering the term 
\begin{equation}
\tilde K(x,x_0)=K_0+K\dn(-\ln T)\dn\left(-\ln \left\|\frac{x}{\|x_0\|}\right\|_{\dn}\right)
\end{equation}
as a state-dependent gain of the  feedback $u=K(x,x_0)x$ we conclude that $\tilde K\to K_0$ as
$x_0\to \zero$. Since the matrix $A+BK_0$ is nilpotent then the stability margin\footnote{The stability margin of a Hurwitz matrix $M\in \R^{n\times n}$ is the maximum of real parts of eigenvalues of $M$.} of the matrix $A+B\tilde K(x,x_0)$  tends to zero as  $x_0\to \zero$. This badly impacts the robustness of the system. Indeed, for $x_0=\zero$ an  additive perturbation of arbitrary small magnitude may invoke unstable motion of the closed-loop system.
To eliminate this drawback, we modify the  feedback law \eqref{eq:PT_hom_control}   as follows
\begin{equation}\label{eq:PT_hom_control_robust}
	u_{\rm ct}(x,x_0)=K_0x+K\dn(-\ln T)\dn\left(-\ln \min\left\{ 1, \left\|\tfrac{x}{\|x_0\|}\right\|_{\dn}\right\}\right)x \quad \text{ for }\quad  x_0\neq \zero	
\end{equation} 
with 
\begin{equation}
    u_{ct}(x,\zero)=u_{\rm lin}(x):=(K_0+K\dn(-\ln T))x \quad \text{ for } \quad x_0=\zero.
\end{equation} Obviously, for $x\neq \zero$, it holds
\begin{equation}
	u_{\rm ct}(x,x_0)\to u_{\rm lin}(x) \quad \text{ as } \quad x_0\to \zero.
\end{equation}
Since $u_{\rm ct}(x,x_0)=u_{\rm lin}(x)$ for $\|x\|\geq \|x_0\|$, the latter limit is uniform on any compact from $ \R^{n}_{\zero}$ and $u_{\rm ct}\in C(\R^n_{\zero}\times  \R^n,\R^m)$.

Such a modification of the  controller allows us to improve the robustness  of the feedback law \eqref{eq:PT_hom_control}
preserving all stability properties of the closed-loop system. 
\begin{theorem}\label{thm:PT_hom_control_robust}
	Let the parameters of the control \eqref{eq:PT_hom_control} be defined as in Theorem \ref{thm:PT_hom_control}. Then the closed-loop system \eqref{eq:mainsystem}, \eqref{eq:PT_hom_control_robust} is fixed-time stable with the constant settling time $T>0$ and it has unique solutions for all $x_0\in \R^n$. All solutions of the perturbed closed-loop system \eqref{eq:mainsystem}, \eqref{eq:PT_hom_control_robust} 
	\begin{equation}\label{eq:system_q}
		\dot x\!\stackrel{a.e.}{=}\!Ax+Bu_{\rm ct}(x+q_1,x_0)+q_2,\quad t\!>\!t_0, \quad x(t_0)\!=\!x_0,  \quad q\!=\!(q_1,q_2)\!\in\! L^{\infty}(\R,\R^{2k}),
	\end{equation}	
	admit the ISS-like estimate
	\begin{equation}\label{eq:ISS_like_estimate}
		\|x(t)\|\leq \|x_0\|\cdot \beta(1,t-t_0)+\|x_0\|\gamma_1\left(\tfrac{\left\| q_1\right\|_{L^{\infty}_{(t_0,t)}}}{\|x_0\|}\right) +
		\|x_0\|\gamma_2\left(\tfrac{\left\| q_2\right\|_{L^{\infty}_{(t_0,t)}}}{\|x_0\|}\right) , \quad  x_0\neq \zero 
	\end{equation}
	where  $\beta\in \mathcal{KL}$ and $\gamma_1,\gamma_2\in\mathcal{K}$ are  independent of $x_0$, but dependent on $T>0$. Moreover 
	\begin{itemize}
		\item there exists $C_i>0$ (dependent of $T>0$) such that 
		\begin{equation}\label{eq:limit_gains}
			\|x_0\|\gamma_i\left(\tfrac{\left\| q_i\right\|_{L^{\infty}_{(t_0,t)}}}{\|x_0\|}\right)\to C_i\left\| q_i\right\|_{L^{\infty}_{(t_0,t)}} \quad \text{ as } \quad \|x_0\|\to 0,
		\end{equation}
		where $i=1,2$;
		\item 	the system  \eqref{eq:system_q} is fixed-time stable  with the settling time estimate $T_{\max}=\frac{\rho T}{\rho-1}$ for $\rho>1$ provided that  $q_1=\zero$, $q_2=B\gamma$, 
		\begin{equation}\label{eq:pt_gamma}
			\|B\gamma(t)\|\leq \|x_0\| \frac{\lambda_{\min}(X^{-1/2}G_{\dn}X^{1/2}+X^{1/2}G_{\dn}^{\top}X^{-1/2})}{2\rho T}, \quad \forall t\in \R.
		\end{equation}
	\end{itemize}
\end{theorem}
\begin{proof}
     The identity $\frac{d}{dt}\|x(t)/\|x_0\||_{\dn}=-\frac{1}{T}$ 
     proven in Theorem \ref{thm:PT_hom_control}  holds also for the unperturbed closed-loop system \eqref{eq:mainsystem}, \eqref{eq:PT_hom_control} since $\|x(t)/\|x_0\|\|_{\dn}\leq 1$ for all $t\geq 0$. So, all conclusions of Theorem \ref{thm:PT_hom_control} remain valid.

	Making the change of variables $z=x/\|x_0\|$ we derive
	\begin{equation}\label{eq:z_system}
		\left\{
		\begin{array}{l}
			\dot z\stackrel{a.e.}{=}Az+B\left(K_0\left(z+q_1^{x_0}\right)+K_T\dn\left(-\ln \min\{\left\|z+q_1^{x_0}\right\|_{\dn},1\}\right)\left(z+q_1^{x_0}\right)\right)+q_2^{x_0}, \\
			z(0)=z_0,
		\end{array}
		\right.
	\end{equation}
	where $q_{1,x_0}=\tfrac{q_1}{\|x_0\|}, q_{2,x_0}=\tfrac{q_2}{\|x_0\|}$,  $K_T=K\dn(-\ln T)$ and 
	$z_0=\tfrac{x_0}{\|x_0\|}$.
	The system \eqref{eq:z_system} is homogeneous in the bi-limit \cite{Andrieu_etal2008:SIAM_JCO}(
	with the zero degree in $\infty$-limit and negative degree in $0$-limit. The unperturbed system \eqref{eq:z_system}
	as well as its homogeneous approximations are globally asymptotically stable. The latter implies ISS with respect to $q_1^{x_0}$ and $q_2^{x_0}$ in the view of the results \cite{Andrieu_etal2008:SIAM_JCO}. Hence, we derive the estimate 
	\eqref{eq:ISS_like_estimate}. 
	By construction, the  stabilizer \eqref{eq:PT_hom_control_robust} tends to the exponentially stabilizing linear feedback $u_{\rm lin}(x)=(K_0+K\dn(-T))x$ as $x_0\to \zero$, but, $x\mapsto \|x\|$ is the Lyapunov function of the corresponding  linear system. Moreover, $u_{\rm ct }(x+q_1,x_0)=u_{lin}(x+q_1)$ if $\|x+q_1\|>\|x_0\|$. The latter implies \eqref{eq:limit_gains}.
	
	Finally, assuming $q_1=\zero$ and $q_2=B\gamma$ we derive $q_1^{x_0}=\zero$ and $q_2^{x_0}=B\frac{\gamma}{\|x_0\|}$. In this case, repeating the proof of Theorem \ref{thm:PT_hom_control} we derive
	\begin{equation}
	\frac{d}{dt}\|z\|_{\dn}\stackrel{a.e.}{=}-\frac{1}{T}+\frac{z^{\top}\dn^{\top}(-\ln \|z\|_{\dn})\dn^{\top}(-\ln T)X^{-1}\dn(-\ln T)B\frac{\gamma}{x_0}}{
		z^{\top}\dn^{\top}(-\ln \|z\|_{\dn})\dn^{\top}(-\ln T)X^{-1}G_{\dn}\dn(-\ln T)\dn(-\ln \|z\|_{\dn})z}
	\end{equation}
	Since $\|\dn(-\ln \|z\|_{\dn})z\|=1$ or, equivalently,
	\begin{equation}
	z^{\top}\dn^{\top}(-\ln \|z\|_{\dn})\dn^{\top}(-\ln T)X^{-1}\dn(-\ln T)\dn(-\ln \|z\|_{\dn})z=1
	\end{equation} then 
	\begin{equation}
	\frac{d}{dt}\|z\|_{\dn}\stackrel{a.e.}{\leq } -\tfrac{1}{T}+\tfrac{2 \|B\gamma\|}{\|x_0\|\lambda_{\min}(X^{-1/2}G_{\dn}X^{1/2}+X^{1/2}G_{\dn}^{\top}X^{-1/2})}\leq -\tfrac{1}{T}+\tfrac{1}{\rho T}=-\tfrac{\rho-1}{\rho T}.
	\end{equation} 
	Taking into account $\|z(0)\|_{\dn}=\left\|\frac{x(0)}{x_0}\right\|_{\dn}=1$ we derive $z(t)=\zero$ for $t\geq T_{\max}$.	
\end{proof}

 \begin{remark}\label{rem:hatx_0}
    If the measurement of the state $x_0$ at the initial instant of time $t=t_0$ is noised as well
    \begin{equation}
    \hat x_0=x_0+q_0, \quad q_0\in \R^n 
    \end{equation}
    then the estimate \eqref{eq:ISS_like_estimate} remains valid 
    by replacing $x_0$ with $\hat x_0$. 
\end{remark}

For small initial conditions the robustness properties  of the closed-system \eqref{eq:mainsystem} with the non-linear feedback \eqref{eq:PT_hom_control_robust} are close to the system with the linear feedback $u_{\rm lin}$.
However, the time-varying  controller \cite{Song_etal2017:Aut} and the fixed-time controller \cite{Polyakov2020:Book} are known to be efficient in rejection of the matched additive disturbances. 
The estimate \eqref{eq:pt_gamma} shows that the static prescribed-time controller \eqref{eq:PT_hom_control_robust} rejects the matched perturbation $\gamma$ of a magnitude proportional $\|x_0\|$.

The further modification of the obtained feedback law 
\begin{equation}\label{eq:FxT_hom_control}
	u_{\rm fxt} (x,x_0)\!=\!K_0x+K\dn(-\ln T)\dn\left(-\ln \min\left\{ 1, \left\|\tfrac{x}{\max\{\|x_0\|,1\}}\right\|_{\dn}\right\}\right)x \quad \text{for}\quad  x_0\!\neq\! \zero,	
\end{equation} 
allows us to enlarge a class of matched perturbations to be rejected, but it guarantees just the fixed-time stabilization (see Definition \ref{def:stab}) with the prescribed upper bound $T_{\max}=T>0$ of the settling time. It is worth stressing that the settling-time estimate  is exact for $\|x_0\|\geq 1$.
\begin{theorem}\label{thm:FxT_hom_control} Let the parameters of the control \eqref{eq:PT_hom_control} be defined as in Theorem \ref{thm:PT_hom_control}.
	The closed-loop system \eqref{eq:mainsystem}, \eqref{eq:FxT_hom_control} is fixed-time stable such that 
	the settling-time  function admits the representation 
	\begin{equation}
	\tilde T(x_0)=\left\{
	\begin{array}{lcc}
		T & \text{ if } & \|x_0\|\geq 1,\\
		T\|x_0\|_{\dn} & \text{ if }  & \|x_0\|\leq 1.
	\end{array}
	\right.
	\end{equation}
	Moreover, the perturbed closed-loop system \eqref{eq:system_q}, \eqref{eq:FxT_hom_control} is 
	\begin{itemize}
		\item ISS with respect to the additive  measurement noises $q_1\in L^{\infty}(\R,\R^n)$ and additive exogenous perturbations $q_2\in L^{\infty}(\R,\R^n)$
		\item fixed-time stable  with $T_{\max}=\frac{\rho T}{\rho-1}$ provided that $\rho>1$, $q_1=\zero$, $q_2=B\gamma$, 
		\begin{equation}
		\|B\gamma(t)\|\leq \max\{1,\|x_0\|\}\frac{\lambda_{\min}(X^{-1/2}G_{\dn}X^{1/2}+X^{1/2}G_{\dn}^{\top}X^{-1/2})}{2\rho T}, \quad \forall t\in \R.
		\end{equation}
	\end{itemize}
\end{theorem}
The proof repeats the proofs of Theorems \ref{thm:PT_hom_control} and  \ref{thm:PT_hom_control_robust}.

\subsection{Predictor-based stabilization of input delay LTI plant}
{
The so-called predictor-based approach \cite{ManitiusOlbrot1979:TAC}, \cite{Artstein1982:TAC},  \cite{Krstic2008:Aut}, \cite{Krstic2009:Book}, \cite{KarafyllisKrstic2018:Book} allows the 
delay-free control design ideas to be extended to input delay systems. The fixed-time stabilizer for the linear generalized homogeneous plant (the integrator chain) with input-delay has been proposed in \cite{Zekraoui_etal2023:Aut} based on the  technique developed in \cite{Krstic2009:Book}, \cite{KarafyllisKrstic2018:Book}, which consist in the modeling of the input delay using a transport PDE (Partial Differential Equation). The control design based on PDE models has one technical limitation: the theory of partial differential equations is not supported with a well-established common  methodology for analysis and design of control systems with state-dependent discontinuities such as to Filippov's method \cite{Filippov1988:Book} for discontinuous ODEs and sliding mode control system \cite{Utkin1992:Book}. 
Some ideas for possible expansion of the sliding mode (discontinuous) control methodology to infinite dimensional system can be found in \cite{Orlov2020:Book}. However, this technique is far to be universal, well-recognized and easy-to-use.
Since our controller \eqref{eq:PT_hom_control_robust}  has the discontinuity at the origin, then the PDE-based design of fixed-time input delay controller is expected to be complicated. However, for linear plants, the predictor-based control design can be easily done using the well-known Artstein's  transformation \cite{Artstein1982:TAC}:
\begin{equation}\label{eq:predictor}
	y(t)=e^{A \tau}x(t)+\int^0_{-\tau}e^{-A\theta}Bu(t+\theta) d\theta, \quad t>0,
\end{equation}
where $\tau>0$ is the input delay. The variable $y$ is the so-called predictor variable, since it estimates the future state $x(t+h)=y(t)$ of the system \eqref{eq:mainsystem}. Notice that if $u\in L^{\infty}([(-\tau,+\infty),\R^m)$ then 
\[
\begin{split}
\dot y(t)\stackrel{a.e.}{=}&e^{A\tau} \dot x(t)+\tfrac{d}{dt}\left(e^{A t} \int^t_{t-\tau}e^{-A\sigma}Bu(\sigma) d\sigma\right)\\
\stackrel{a.e.}{=}&
e^{A\tau} (Ax(t)\!+\!Bu(t-\tau))\!+\!Ae^{A t} \int^t_{t-\tau}e^{-A\sigma}Bu(\sigma) d\sigma
\!+\!e^{A t}\!\tfrac{d}{dt}\!\int^t_{t-\tau}\!\!\!\!\!\!\!\!e^{-A\sigma}Bu(\sigma) d\sigma
\\
\stackrel{a.e.}{=}&Ay(t)+e^{A\tau}Bu(t-\tau)+e^{At}e^{-At}Bu(t)-e^{At}e^{-A(t-\tau)}Bu(t-\tau).
\end{split}
\]
Therefore, the dynamics of the predictor variable is governed by the ODE
\begin{equation} \label{eq:system_y}
\dot y(t)\stackrel{a.e.}{=}Ay(t)+Bu(t), \quad t>0, \quad y(0)=y_0,
\end{equation}
where $y_0=e^{A\tau}x_0+\int^0_{-\tau}e^{-A\sigma}\phi(\sigma)d\sigma$ and $\phi\in L^{\infty}((-\tau,0),\R^m)$ defines the control signal $u$ (see, \eqref{eq:initial_condition}) on the time interval $[-\tau,0]$. 

The Artstein's transformation reduces the problem of a control design for the input delay system to the same problem in the delay free case.
By Theorem \ref{thm:PT_hom_control_robust}, the fixed-time stabilizer for the delay-free system \eqref{eq:system_y} can be designed in the form of the discontinuous feedback \eqref{eq:PT_hom_control_robust}. 
Following the Filippov regularization technique (see Lemma \ref{lem:wp_delay_free}) we define
\begin{equation}\label{eq:PT_hom_control_y}
   u(t)\stackrel{a.e.}{\in} \bar{u}_{\rm ct}(y(t),y_0):=\left
\{
\begin{array}{ccc}
 u_{\rm ct}(y(t),y_0) & \text{ if } & y(t)\neq\zero,\\
 \|y_0\|K\dn(-\ln T) \mathcal{B}& \text{ if } & y(t)=\zero,
\end{array}
\right.   
\end{equation}
where $u_{\rm ct}$ is given by \eqref{eq:PT_hom_control_robust}
and $\mathcal{B}=\{x\in \R^n: \|x\|\leq 1\}$ is the unit ball.
\begin{lemma}[Well-posedness of the input delay control system] 
	For any $x_0\neq \zero$ and any $\phi\in L^{\infty}((-\tau,0),\R^m)$, there exists a tuple $$(x,y,u)\in W^{1,\infty}_{\rm loc}(\R_+,\R^n) \times  W^{1,\infty}_{\rm loc}(\R_+,\R^n) \times  L^{\infty}_{\rm loc}([-\tau,+\infty),\R^m)$$ satisfying \eqref{eq:mainsystem}, \eqref{eq:initial_condition}, \eqref{eq:PT_hom_control_y}, \eqref{eq:predictor}.

\end{lemma}
\begin{proof} 
By Lemma \ref{lem:wp_delay_free} the system \eqref{eq:system_y},\eqref{eq:PT_hom_control_y} is well-posed and
has a Filippov solution $y:\R_+\mapsto \R^n$. The  Filippov's lemma \cite{Filippov1962:SIAM} about measurable selector guarantees that 
there exists a measurable function  $u\in L^{\infty}_{\rm loc}(\R_+,\R^m)$ such that 
\begin{equation}
  \left\{
  \begin{array}{l}
  \dot y(t)\stackrel{a.e.}{=}Ay(t)+Bu(t), \\
  u(t)\stackrel{a.e.}{\in} \bar{u}_{\rm ct}(y(t),y_0),
  \end{array}
  \right.\quad t>0.
\end{equation}

Let us extend the signal $u$ to the time interval $(-\tau,+\infty)$ using the initial condition \eqref{eq:initial_condition}.
Applying the inverse Artstein's tranformation  $$x(t)=e^{-A\tau}\left(y(t)-\int^{0}_{-\tau}
e^{-A\sigma}u(t+\sigma) d\sigma)\right)$$ 
 we derive  $x(0)=x_0$ and 
 \[
 \dot x(t)=e^{-A\tau }\left(\dot y-\tfrac{d}{dt}\left(e^{A t} \int^t_{t-\tau}e^{-A\sigma}Bu(\sigma) d\sigma\right)\right)=
 \]
 \[
 e^{-A\tau }\left(Ay(t)+Bu(t)-\tfrac{d}{dt}\left(e^{A t} \int^t_{t-\tau}e^{-A\sigma}Bu(\sigma) d\sigma\right)\right)=Ax(t)+Bu(t-\tau).
 \]
 Therefore, the the constructed tuple $(x,y,u)$ satisfies 
\eqref{eq:mainsystem}, \eqref{eq:initial_condition}, \eqref{eq:PT_hom_control_y}, \eqref{eq:predictor}.
\end{proof}

The following theorem proves that the closed-loop input-delay system  is fixed-time stable with a constant settling time.
\begin{theorem}\label{thm:PT_predictor}
If all parameters of the control \eqref{eq:PT_hom_control_y} are defined as in Theorem \ref{thm:PT_hom_control} and $\phi=\zero$ then the system \eqref{eq:mainsystem} with the control signal generated by the formulas \eqref{eq:initial_condition}, \eqref{eq:PT_hom_control_y}, \eqref{eq:predictor} is fixed-time stable (in the sense of Definition \ref{def:stab}) with the constant settling time $T+\tau$ , namely,  $\exists \alpha \in \mathcal{K} : \|x(t)\|\leq \alpha(\|x_0\|),\forall t\geq 0$ and 
 \begin{equation}
     x(t)=\zero,\quad \forall t\geq T+\tau
 \end{equation}
 independently of $x_0\in \R^n$, but $x(t)\neq \zero$ for all $t\in[0,T+\tau)$ provided that $x_0\neq \zero$.
\end{theorem}
\begin{proof} Let $(x,y,u)\in W^{1,\infty}_{\rm loc}(\R_+,\R^n) \times  W^{1,\infty}_{\rm loc}(\R_+,\R^n) \times  L^{\infty}_{\rm loc}([-\tau,+\infty),\R^m)$ satisfy  \eqref{eq:mainsystem},\eqref{eq:initial_condition}, \eqref{eq:predictor}, \eqref{eq:PT_hom_control_y}.
Then, due to Artstein's transformation, $y$ satisfies the following differential inclusion

\begin{equation} \label{eq:system_y2}
\dot y(t)\stackrel{a.e.}{\in } Ay(t)+B\overline u_{ct}(y(t),y_0), \quad t>0, \quad y(0)=y_0.
\end{equation}
Since, by Theorem \ref{thm:PT_hom_control_robust}, the latter system is fixed-time stable with a constant settling time $T>0$ and has the unique solution: 
\[
y(t)=\zero \quad \text{ for all } t\geq T
\]
and $y(t)\neq \zero$ for all $t\in[0,T)$ and all $y_0\in \R^n$.
Moreover, since $\phi=\zero$ then $y_0=e^{A\tau }x_0$ and $x_0\neq \zero \Rightarrow y_0\neq \zero$, so  $y(t)\neq \zero$
for all $t\in [0,T)$
independently of $x_0\in \R^n_\zero$.
Since $y(t)=\dot y(t)=\zero,\forall t> T$ then, in the view of the equation \eqref{eq:system_y}, we have $Bu(t)\stackrel{a.e.}{=}\zero, t>T$. Therefore,
\begin{equation}
\left\{
\begin{array}{l}
	\zero\stackrel{a.e.}{=}Bu(t),\\
	\zero=e^{A\tau} x(t)+\int\limits^0_{-\tau} e^{-A\theta} Bu(t+\theta) d\theta, 
\end{array}
\right. \quad \forall t> T \quad \Rightarrow \quad x(t)=\zero, \quad \forall t\geq T+\tau.
\end{equation}
The identity $x(t)=y(t-\tau)$ holds for all $t\geq \tau$. 
To prove the constant convergence time, we just need to show
that $x_0\neq \zero \Rightarrow x(t)\neq \zero$ for all $t\in [0,\tau]$.
Suppose the contrary, i.e., for some $x_0\neq \zero$ there exists $t^*\in (0,\tau]$ such that 
\begin{equation}
    x(t^*)=\zero.
\end{equation}
On the one hand, since  $x(t^*)=e^{-A\tau} \left(y(t^*)-\int\limits^0_{-\tau} e^{-A\theta} Bu(t^*+\theta) d\theta\right)$ then 
\begin{equation}
 y(t^*)=\int\limits^{t^*}_{0} e^{A(t^*-\sigma)} Bu(\sigma) d\sigma.
\end{equation}
On the other hand, since $\phi=\zero$ and $y$ is satisfies \eqref{eq:system_y2}, then, by Cauchy formula, we have
\begin{equation}
y(t^*)=e^{At^*}y_0+\int\limits^{t^*}_{0} e^{A(t^*-\sigma)} Bu(\sigma) d\sigma.
\end{equation}
Hence, we derive $e^{At^*}y_0=\zero$, but the latter is possible 
if an only if $y_0=\zero$ (or, equivalently, $x_0=\zero$). We derive the contradiction.

Finally, since the system \eqref{eq:system_y2} is Lyapunov stable then there exists $\bar \alpha \in \mathcal{K}$ such that $\|y(t)\|\leq \bar \alpha(\|y_0\|),\forall t\geq 0$.
Taking into account
\begin{equation}
\|u_{\rm ct}(y,y_0)\|_{\R^m}\leq (\|K_0\|+\|K\dn(-\ln T)\|)\cdot  (\|y_0\|+ \|y\|).
\end{equation}
and  using \eqref{eq:predictor}, we derive
\begin{equation}
   \|x(t)\|\leq \|e^{-A\tau}\| \bar \alpha(\|y_0\|)+\int^{0}_{-h} \|e^{-A s}\|
   ds \cdot (\|K_0\|+\|K\dn(-\ln T)\|)\cdot  (\|y_0\|+ \bar \alpha(\|y_0\|).
\end{equation}
Since $y_0=e^{A\tau}x_0$ then there exists $\bar \alpha\in \mathcal{K}$ such that 
$\|x(t)\|\leq\bar \alpha(\|x_0\|),\forall t\geq 0$. The proof is complete.
 \end{proof}
 
Notice that it is impossible to assign the  constant settling time for all $x_0\neq \zero$ if $\phi\neq \zero$. Indeed, since the pair $\{A,B\}$ is controllable then $\int^{0}_{-\tau} e^{-A\theta} B\phi(\theta)d\theta\neq \zero$ for $\phi\neq \zero$. In this case, taking $x_0=-e^{-A\tau}\int^{0}_{-\tau} e^{-A\theta} B\phi(\theta)d\theta\neq \zero$ we derive  that the unique solution of  the closed-loop system satisfies $x(\tau)=y_0=\zero$. The latter implies that $x(t)=\zero$ for all $t\geq \tau$. 
However, if $x_0=-2e^{-A\tau}\int^{0}_{-\tau} e^{-A\theta} B\phi(\theta)d\theta\neq \zero$ and $x(\tau)\neq \zero$ for any  fixed-time stabilizing controller. Therefore, at least, the two different non-zero initial vectors $x_0\neq \zero$ corresponds  two different settling times of the system to zero.

To analyze a robustness with respect to perturbations (such as measurement noises, additive disturbances, computational errors for the predictor variable, etc) we consider the system
\begin{equation}\label{eq:closedloop_delay_q}
\left\{
\begin{array}{l}
\dot x(t)=Ax(t)+Bu(t-\tau)+q_2(t),\quad t>0,\quad x(0)=x_0,\\
y(t)=e^{A \tau}x(t)+\int^0_{-\tau}e^{-A\theta}Bu(t+\theta) d\theta,\\
u(t)\in \bar{u}_{\rm ct}(y(t)+q_1(t),y_0),
\end{array}
\right.
\end{equation}
where $\overline{u}_{\rm ct}$ is given above, $q_1\in L^{\infty}(\R,\R^n)$ models measurement noises and computational error for the predictor variable $y$, but $q_2\in L^{\infty}(\R,\R^n)$ models additive perturbations of the plant.
\begin{theorem}
Under the conditions of Theorem \ref{thm:PT_predictor}, any solution of the system \eqref{eq:closedloop_delay_q} satisfies the ISS-like estimate
	\begin{equation}\label{eq:ISS_like_estimate_delay}
		\|x(t)\|\leq \|e^{A\tau}x_0\|\cdot \beta(1,t-t_0)+\|e^{A\tau}x_0\|\cdot \gamma_1\left(\tfrac{\left\| q_1\right\|_{L^{\infty}_{(t_0,t)}}}{\|e^{A\tau}x_0\|}\right) +
		\|e^{A\tau}x_0\|\cdot \gamma_2\left(\tfrac{\left\| q_2\right\|_{L^{\infty}_{(t_0,t)}}}{\|e^{A\tau}x_0\|}\right), 
	\end{equation}
	where  $\beta\in \mathcal{KL}$ and $\gamma_1,\gamma_2\in\mathcal{K}$ are  independent of $x_0$, but depend on $T>0$ and $\tau>0$. 
\end{theorem}
\begin{proof}
    For $\phi=\zero$, the predictor dynamics is described by
    \begin{equation}\label{eq:system_y_pert}
    \left\{
    \begin{array}{l}
    \dot y(t)=Ay(t)+Bu(t)+e^{A\tau}q_2(t), \quad t>0, \quad y(0)=e^{A\tau}x_0,\\
    u(t)\in \bar{u}_{\rm ct}(y(t)+q_1(t),e^{A\tau}x_0).
    \end{array}
    \right.
    \end{equation}
   The ISS of the latter system is studied in Theorem \ref{thm:PT_hom_control_robust}, where it is shown that there exist $\beta^y\in\mathcal{KL}, \gamma_1^y,\gamma_2^y\in \mathcal{K}$
   such that
   \begin{equation}\label{eq:ISS_y}
   \|y(t)\|\!\leq\! \|e^{A\tau}x_0\|\cdot \beta^y(1,t-t_0)+\|e^{A\tau}x_0\|\cdot \gamma_1^y\!\left(\tfrac{\left\| q_1\right\|_{L^{\infty}_{(t_0,t)}}}{\|e^{A\tau}x_0\|}\right)\!+
		\|e^{A\tau}x_0\|\cdot \gamma_2^y\!\left(\tfrac{\left\| q_2\right\|_{L^{\infty}_{(t_0,t)}}}{\|e^{A\tau}x_0\|}\right)\!.
   \end{equation}
   On the other hand, by Cauchy formula, we have
\begin{equation}
y(t)=e^{A\tau}y(t-\tau)+\int^{0}_{-\tau} e^{-A\theta}Bu(t+\theta)d\theta +\int^{0}_{-\tau}e^{-A\theta} q_2(t+\theta)d\theta, \quad t\geq \tau,
\end{equation}
so
\begin{equation}
\left\|\int^{0}_{-\tau} \!\!\!\!e^{-A\theta}Bu(t+\theta)d\theta\right\|\!\leq\! \|y(t)\|+\|e^{A\tau}\| \left(\|y(t-\tau)\|
+\|q_2\|_{L^{\infty}_{(t_0,t)}}\tau\!\sup_{\theta\in [0,\tau]}\|e^{A\theta}\|\right)\!. 
\end{equation}
On the one hand, using the formula \eqref{eq:predictor} we derive
\begin{equation}
\|x(t)\|\leq \|e^{-A\tau}\|\left(\|y(t)\|+\left\|\int^{0}_{-\tau} e^{-A\theta}Bu(t+\theta)d\theta\right\|\right)\leq 
\end{equation}
\begin{equation}
\|e^{-A\tau}\|\left(2\|y(t)\|+\|e^{A\tau}\| \left(\|y(t-\tau)\|+\|q_2\|_{L^{\infty}_{(t_0,t)}}\tau \sup_{\theta\in [0,\tau]}\|e^{A\theta}\|\right)\right),\quad \forall t\geq \tau.
\end{equation}
On the other hand, since $\|u(t)\|\leq (\|K_0\|+\|K\dn(-\ln T)\|)(\|y(t)\|+\|e^{A\tau}x_0\|), \forall t\geq 0$ then, taking into account $\phi=\zero$, by the Artstein's transformation, we derive 
\begin{equation}
\|x(t)\|\leq \|e^{-A\tau}\| \left(\|y(t)\|+ C_1\|e^{A\tau}x_0\|+C_2\sup_{\sigma\in [0,t]}\|y(\sigma)\|\right), \quad \forall t\in [0,\tau].
\end{equation}
for some $C_1,C_2>0$.
Therefore, using \eqref{eq:ISS_y} we derive that 
the ISS-like estimate \eqref{eq:ISS_like_estimate_delay} holds for $t\geq 0$ under a properly defined functions $\beta\in \mathcal{KL}$ and $\gamma_1,\gamma_2\in \mathcal{K}$.
\end{proof}

The matched perturbation $q_2=B\gamma$ becomes mismatched for the predictor system \eqref{eq:system_y_pert}, so it cannot be completely rejected as it was done in the delay-free case. Therefore, in the input delay case, the robustness of the prescribed-time stabilizers  with respect to additive perturbations is proven only in the ISS-like sense. Notice that, the gains $\|x_0\|\cdot \gamma_i(\|q_i\|/\|x_0\|), i=1,2$ tend to some linear functions of $\|q_i\|$  as $\|x_0\|\to 0$ similarly to the delay free case (see, the formula \eqref{eq:limit_gains}). However, these gains may depend now on both the parameter $T>0$, which prescribes the regulation time, and  the input delay $\tau>0$ of the system.  Notice that  Remark \ref{rem:hatx_0}
is valid for the input delay case and the estimate \eqref{eq:ISS_like_estimate_delay}.
}

\section{Numerical Example}
\subsection{Prescribed-time stabilization of the harmonic oscillator in the delay-free case}
As an example, let us design the prescribed-time stabilizer for the harmonic oscillator in the delay-free ($\tau=0$) case
\begin{equation}
A=\left(
\begin{array}{cc}
	0 & 1\\
	-1  & 0
\end{array}
\right), \quad B=\left(
\begin{array}{cc}
	0\\
	1&
\end{array}
\right)
\end{equation}
The parameters of the prescribed-time stabilizer \eqref{eq:PT_hom_control_robust} are designed according to Theorem \ref{thm:PT_hom_control}:
\begin{equation}
K_0\!=\!(1 \;\; 0), \quad G_{\dn}=\left(\!
\begin{array}{cc}
	2 & 0\\
	0  &1
\end{array}\!
\right), \quad K\!=\!(-5.5  \;\; -3), \quad X\!=\!\left(\!
\begin{array}{cc}
	1 & -2\\
	-2 & 5.5
\end{array}\!
\right), \quad T\!=\!1. 
\end{equation}
The simulation has been done in MATLAB using the zero-order-hold method and the consistent discretization
of the homogeneous controller \eqref{eq:PT_hom_control_robust}  realized in Homogeneous Control  Systems Toolbox\footnote{\texttt{https://gitlab.inria.fr/polyakov/hcs-toolbox-for-matlab}} for MATLAB. 
The consistent discretization (see \cite{Polyakov_etal2023:Aut}) allows the convergence rate (e.g., finite-time or fixed-time convergence) of the continuous-time control system to be preserved in the case of the sampled-time implementation of the controller. The sampling period for the simulation is $h=0.01$. The simulation results 
show the prescribed-time convergence of the closed-loop system with $T=1$. Indeed, independently of the selected initial condition (see Figures \ref{fig:nominal_02} and \ref{fig:nominal_07}) the state of the closed-loop system converges to zero with the precision of the the machine epsilon ($\approx 10^{-16}$) at the time instant $1.01$, which perfectly corresponds to the prescribed settling time $T=1$ (up to the sampling period). The simulations have been done for various initial conditions up to $\|x_0\|=10^5$. The settling time remains equal to $1$ (up to the sampling period $h$) in all simulation and various $h<T$.

\begin{figure}
	\centering
	\includegraphics[width=120mm]{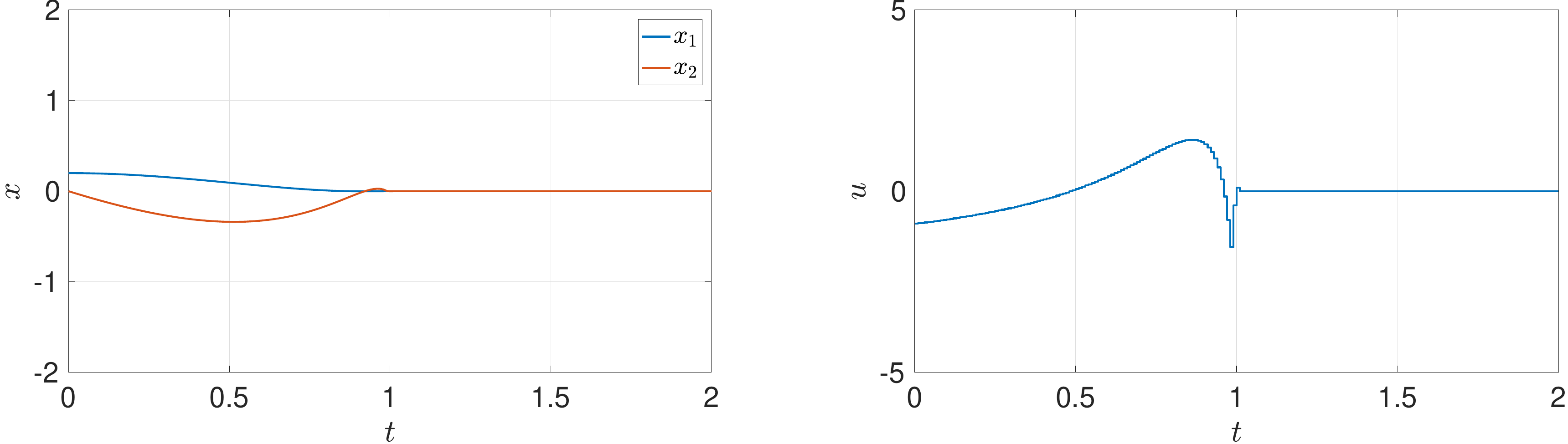}
	\caption{Stabilization at the prescribed-time $T=1$ for $x_0=(0.2 \;\; 0)^{\top}$}
	\label{fig:nominal_02}
\end{figure}

\begin{figure}
\centering
	\includegraphics[width=120mm]{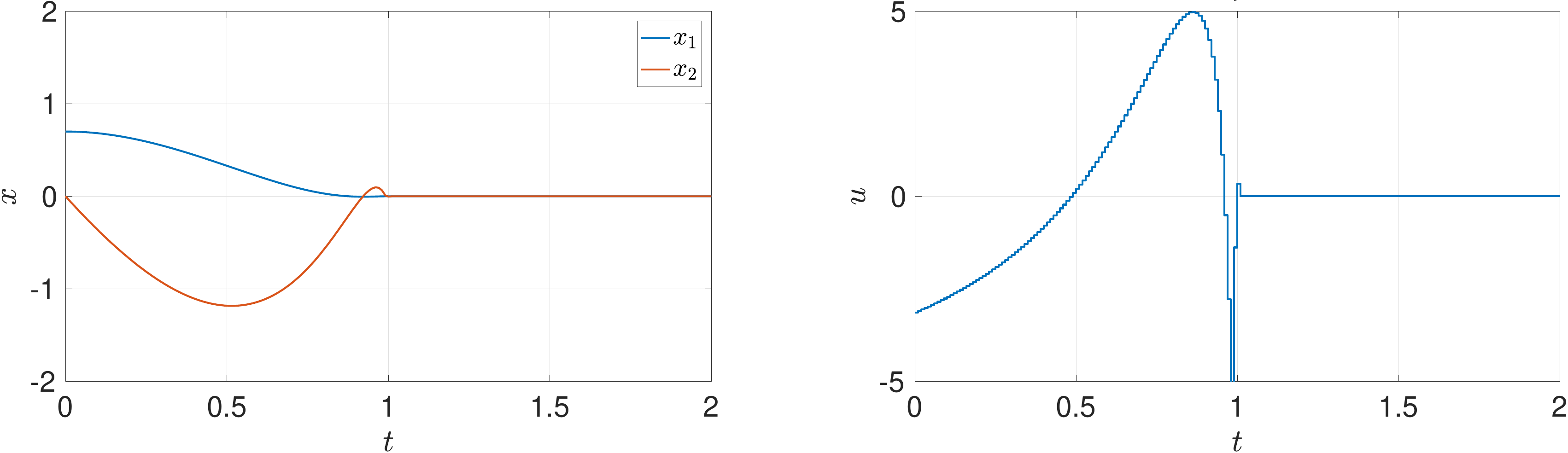}
	\caption{Stabilization at the prescribed-time $T=1$ for $x_0=(0.7 \;\; 0)^{\top}$}
	\label{fig:nominal_07}
\end{figure}

To study robustness properties of the closed-loop system, the simulations have been done, first, for the system with matched additive perturbation $B\sin(5t)$. As claimed in Theorem \ref{thm:PT_hom_control_robust}, such a perturbation cannot be rejected by the prescribed-time controller \eqref{eq:PT_hom_control_robust} if the initial state is to small (see Figure \ref{fig:pert_matched_02}). The larger initial condition, the larger matched perturbation that can be rejected (see Figure \ref{fig:pert_matched_07}). The fixed-time stabilizer \eqref{eq:FxT_hom_control} rejects the the considered matched perturbation for all initial conditions.

\begin{figure}
	\centering
	\includegraphics[width=120mm]{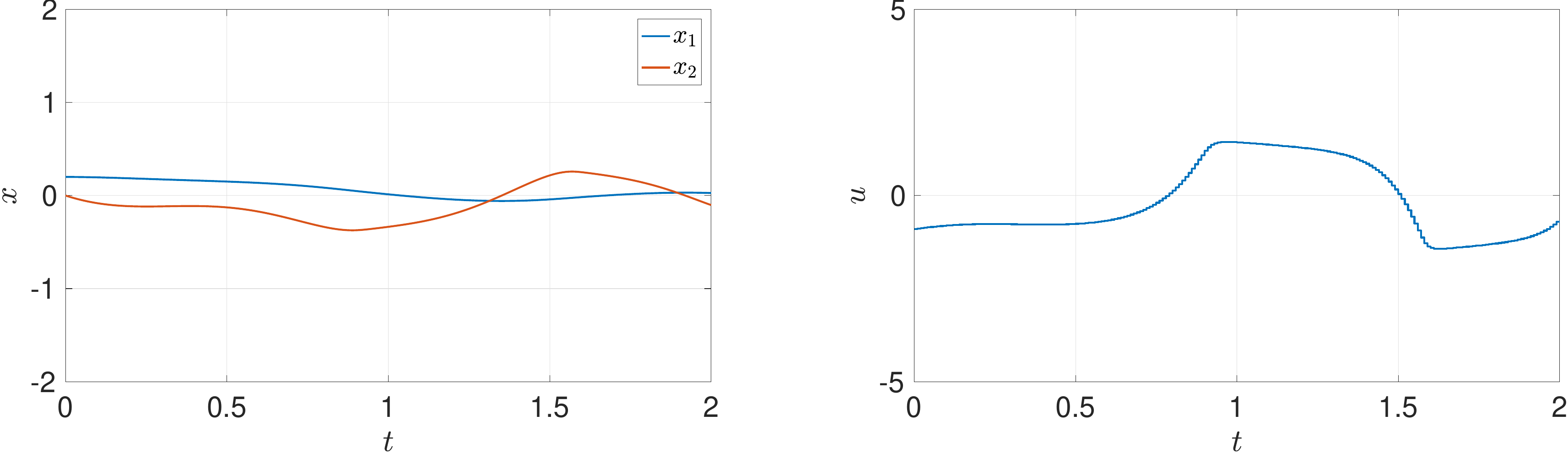}
	\caption{The case of matched  additive disturbance $B\sin(5t)$  for $x_0=(0.2 \;\; 0)^{\top}$}
	\label{fig:pert_matched_02}
\end{figure}
\begin{figure}
\centering
	\includegraphics[width=120mm]{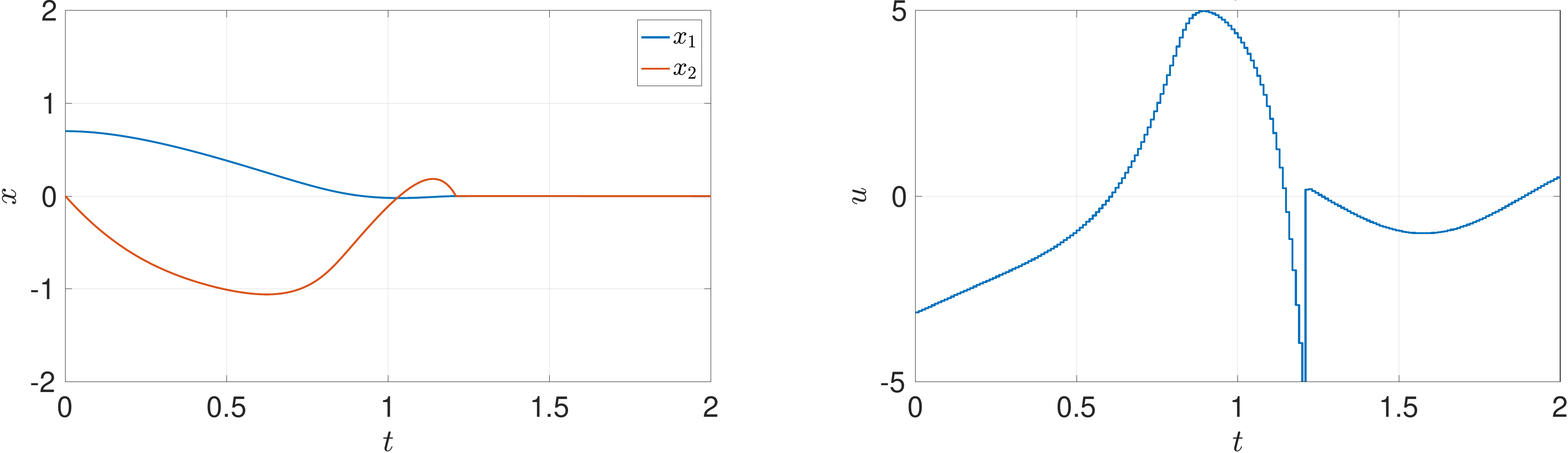}
	\caption{The case of matched additive disturbance  $B\sin(5t)$ for $x_0=(0.7 \;\; 0)^{\top}$}
	\label{fig:pert_matched_07}
\end{figure}
The ISS with respect to noised measurements is quite opposite to the case of ISS with respect to additive perturbations in the sense that the smaller initial state $x_0$, the less sensitive closed-loop system with respect to measurement noises (see Figures \ref{fig:noise_02} and \ref{fig:noise_07}). The numerical simulations for this case have been done by adding a noise $\eta$ of the magnitude $0.01$ to the state measurements $\hat x=x+\eta$. The noise is simulated by MATLAB as a uniformly distributed (pseudo-)random variable $\eta\in [-0.01,0.01]$.
\begin{figure}
	\centering
	\includegraphics[width=120mm]{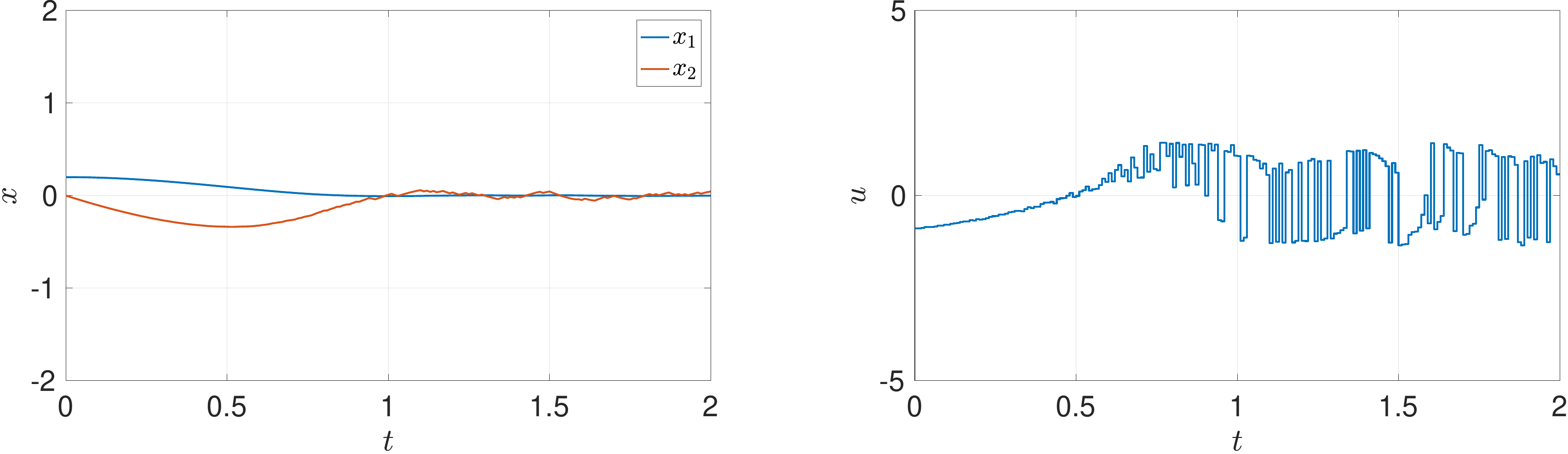}
	\caption{The case of noised measurements   for $x_0=(0.2 \;\; 0)^{\top}$}
	\label{fig:noise_02}
\end{figure}
\begin{figure}
\centering
	\includegraphics[width=120mm]{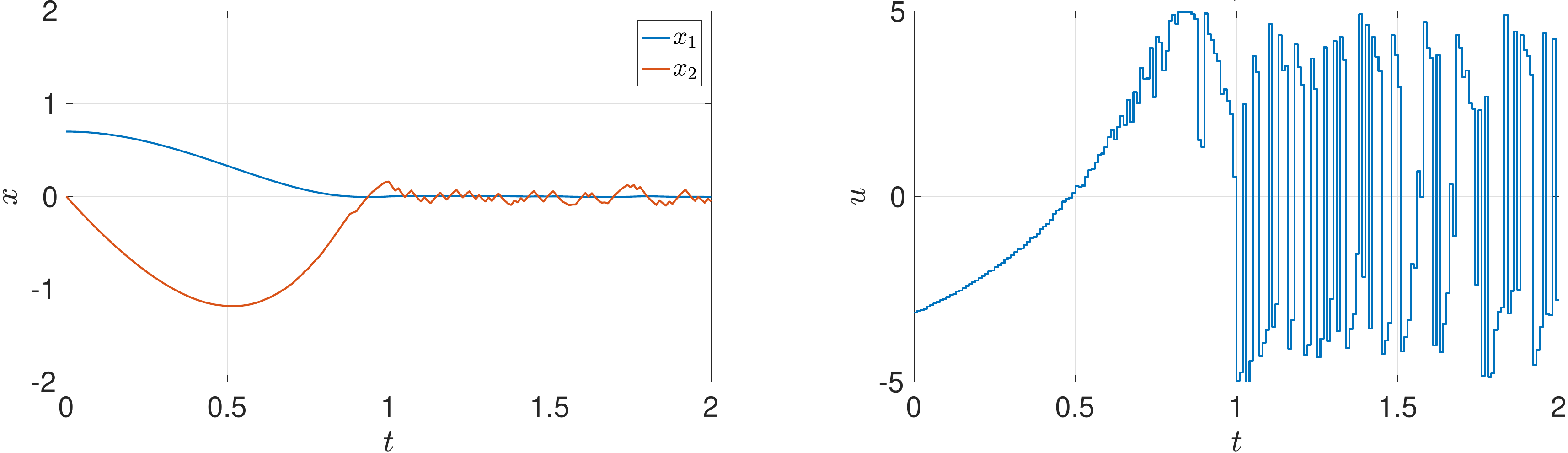}
	\caption{The case of noised measurements   for $x_0=(0.7 \;\; 0)^{\top}$}
	\label{fig:noise_07}
\end{figure}
\subsection{Prescribed-time stabilization of the input delay system}

Let the model of the controller harmonic studied above have the input delay $\tau=0.5$. In this case, the prescribed-time feedback has to be calculated using the predictor variable $y$ given by \eqref{eq:predictor}. To implement the method of consistent discretization\cite{Polyakov_etal2023:Aut} to the system \eqref{eq:system_y}, \eqref{eq:PT_hom_control_y}, the predictor variable has to be calculated exactly at the time instances
\begin{equation}
t_i=ih, \quad i=0,1,\ldots,
\end{equation}
where $h=0.01$ is the sampling period. Since the control signal is a piece-wise constant function with the sampling period $h=0.01$, the integral term in \eqref{eq:predictor} admits the following exact representation
\begin{equation}
\int^0_{-\tau}\!\!\!\!\!e^{-A\theta}u(t_i+\theta) d\theta\!=\!\!\sum_{j=1}^{N}\!\left(\!\int\limits_{jh}^{(j-1)h}\!\!\!\!\!\!e^{A\tilde \theta}d\tilde \theta\!\right)\!u(t_i-jh)\!=\!\!
\sum_{j=1}^{N}\!A^{-1}\!(e^{jhA}-e^{(j-1)hA}) u(t_i-jh),
\end{equation}
where  $N=\frac{\tau}{h}=50$ and  $A^{-1}=-A$ (for our model of the harmonic oscillator). Let the control for the predictor equation \eqref{eq:system_y} be designed
as for the delay-free system considered above. 
Due to the input delay the control signal generated at the time $t$ impacts on the system at the time instant $t+\tau$. The control signal as well as the predictor variable converge to a steady state (e.g., to zero) at the prescribed-time $T=1$, but the expected settling time of the system is $T+\tau=1.5$. The numerical simulations show this prescribed converge time (see, Figures \ref{fig:delay_nominal_02} and \ref{fig:delay_nominal_07}) for the closed-loop system. 

\begin{figure}
	\centering
	\includegraphics[width=120mm]{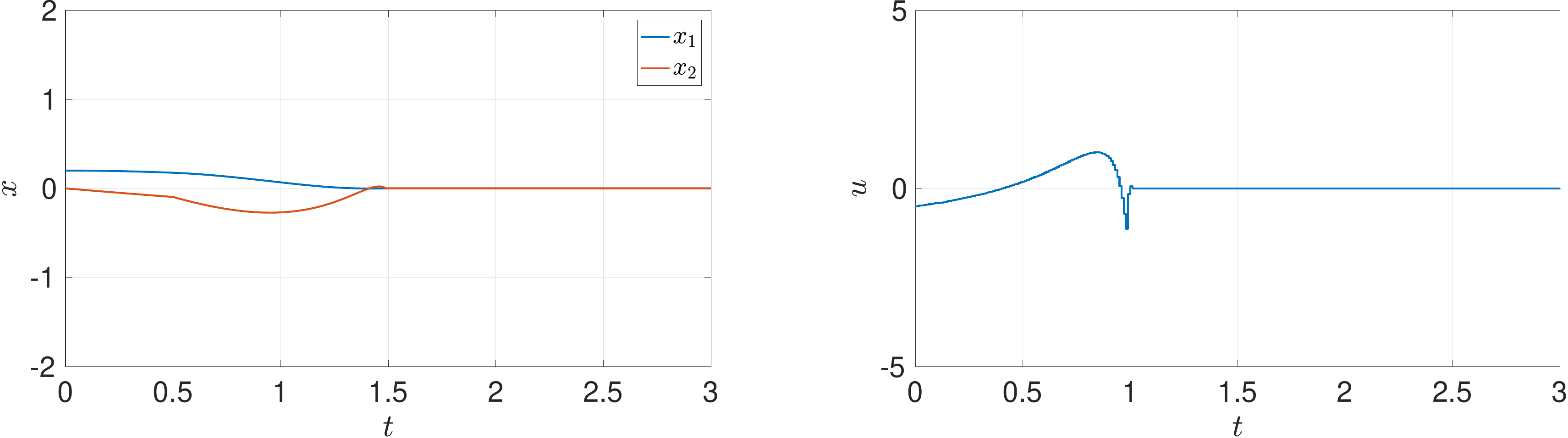}
	\caption{Stabilization at the prescribed-time $T=1.5$ for $x_0=(0.2 \;\; 0)^{\top}$ and the input delay $\tau=0.5$}
	\label{fig:delay_nominal_02}
\end{figure}

\begin{figure}
\centering
	\includegraphics[width=120mm]{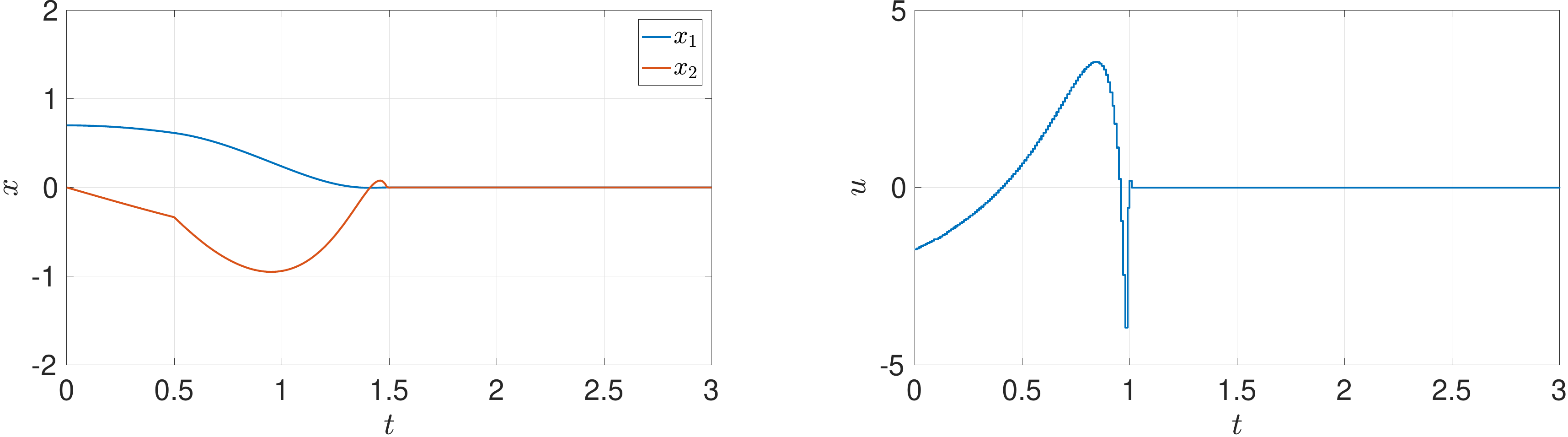}
	\caption{Stabilization at the prescribed-time $T=1.5$ for $x_0=(0.7 \;\; 0)^{\top}$ and the input delay $\tau=0.5$}
	\label{fig:delay_nominal_07}
\end{figure}

Notice that the matched additive perturbation $B\sin(5t)$ of the original system becomes the mismatched additive perturbation $e^{A\tau}B\sin(5t)$ in the predictor equation 
\eqref{eq:system_y_pert}. So, this perturbation cannot be rejected by the predictor-based  stabilizer and just ISS with respect the additive perturbations can be guaranteed (see Figure \ref{fig:delay_pert_matched_02} and \ref{fig:delay_pert_matched_07}). The conclusions about sensitivity with respect to measurement noises are the same as in the delay free case. 

\begin{figure}
	\centering
	\includegraphics[width=120mm]{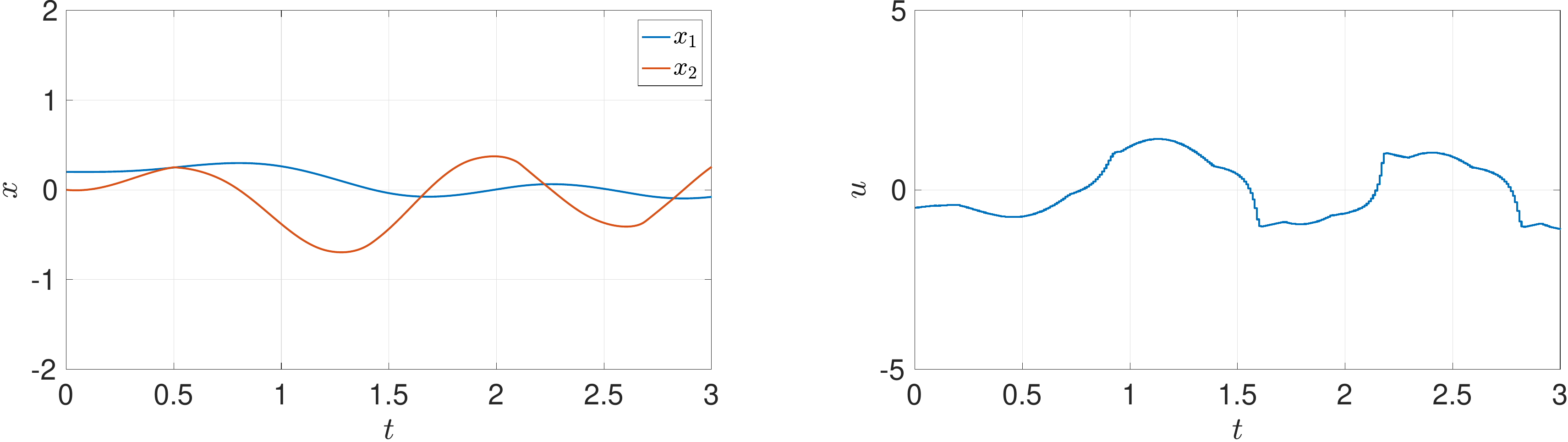}
	\caption{The case of matched  additive disturbance  $B\sin(5t)$ for $x_0=(0.2 \;\; 0)^{\top}$ and the input delay $\tau=0.5$}
	\label{fig:delay_pert_matched_02}
\end{figure}
\begin{figure}
\centering
	\includegraphics[width=120mm]{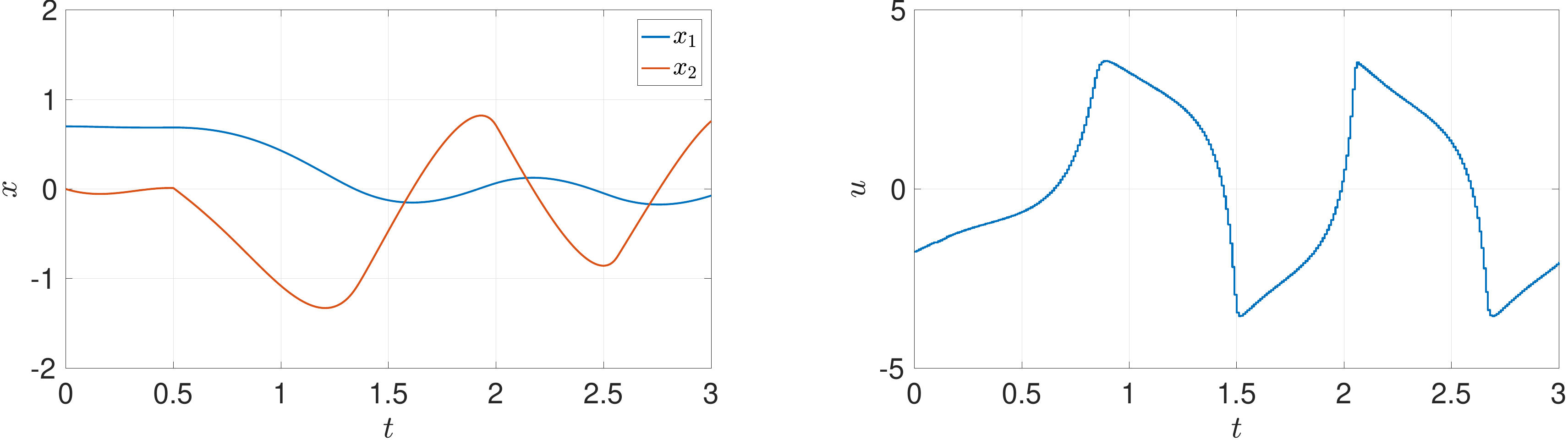}
	\caption{The case of matched additive disturbance  $B\sin(5t)$ for $x_0=(0.7 \;\; 0)^{\top}$ and the input delay $\tau=0.5$}
	\label{fig:delay_pert_matched_07}
\end{figure}



\section{Conclusions}
In this paper, new  fixed-time stabilizers are designed for LTI systems. The key feature of the new stabilizers is the dependence of the feedback gain on the initial condition. This allows the settling time of the closed-loop to have a prescribed constant settling time for all non-zero initial conditions. The obtained stabilizer does not have a time varying gain which tends to infinity as time tends to the settling time. The latter essentially improves the robustness properties of the closed-loop system with respect to measurement noises comparing to well-known time-varying prescribed-time stabilizers (like \cite{Song_etal2017:Aut}. The control laws are designed for both delay-free and input-delay cases. The theoretical results are illustrated by numerical simulations.

\bibliographystyle{plain}
\bibliography{bib_all.bib}

\begin{thebibliography}{10}

\bibitem{Abel_etal2022:ACC}
I.~Abel, D.~Steeves, M.~Krstic, and M.~Jankovic.
\newblock Prescribed-time safety design for a chain of integrators.
\newblock In {\em American Control Conference}, 2022.

\bibitem{Andrieu_etal2008:SIAM_JCO}
V.~Andrieu, L.~Praly, and A.~Astolfi.
\newblock {Homogeneous Approximation, Recursive Observer Design, and Output
  Feedback}.
\newblock {\em SIAM Journal of Control and Optimization}, 47(4):1814--1850,
  2008.

\bibitem{Artstein1982:TAC}
Z.~Artstein.
\newblock Linear systems with delayed controls: A reduction.
\newblock {\em IEEE Transaction on Automatic Control}, 27(4):869--879, 1982.

\bibitem{Balakrishnan2005:AMC}
A.V. Balakrishnan.
\newblock Superstability of systems.
\newblock {\em Applied Mathematics and Computation}, 164:321--326, 2005.

\bibitem{BhatBernstein2005:MCSS}
S.~P. Bhat and D.~S. Bernstein.
\newblock Geometric homogeneity with applications to finite-time stability.
\newblock {\em Mathematics of Control, Signals and Systems}, 17:101--127, 2005.

\bibitem{BhatBernstein2000:SIAM_JCO}
S.P. Bhat and D.S. Bernstein.
\newblock Finite time stability of continuous autonomous systems.
\newblock {\em SIAM J. Control Optim.}, 38(3):751--766, 2000.

\bibitem{EfimovPolyakov2021:Book}
D.~Efimov and A.~Polyakov.
\newblock {\em Finite-time stability tools for control and estimation}.
\newblock Foundations and Trends in Systems and Control. Netherlands: now
  Publishers Inc, 2021.

\bibitem{Espitia_etal2019:Aut}
N.~Espitia, A.~Polyakov, D.~Efimov, and W.~Perruquetti.
\newblock Boundary time-varying feedbacks for fixed-time stabilization of
  constant-parameter reaction--diffusion systems.
\newblock {\em Automatica}, 103(5):398--407, 2019.

\bibitem{Feldbaum1953:AiT}
A.~Feldbaum.
\newblock Optimal processes in systems of automatic control.
\newblock {\em Avtom. Telemekh.}, 14(6):721--728, 1953 (in Russian).

\bibitem{Filippov1962:SIAM}
A.~F. Filippov.
\newblock On certain questions in the theory of optimal control.
\newblock {\em J. SIAM Control}, 1(1):76--84, 1962.

\bibitem{Filippov1988:Book}
A.~F. Filippov.
\newblock {\em Differential Equations with Discontinuous Right-hand Sides}.
\newblock Kluwer Academic Publishers, 1988.

\bibitem{Fuller1960:IFAC}
A.~Fuller.
\newblock Relay control systems optimized for various performance criteria.
\newblock In {\em In Proceedings of the 1st IFAC World Congress}, pages
  510--519, 1960.

\bibitem{GeligLeonovYakubovich2004:Book}
A.~Kh. Gelig, G.~A. Leonov, and V.~A. Yakubovich.
\newblock {\em Stability Of Stationary Sets In Control Systems With
  Discontinuous Nonlinearities}.
\newblock World Scientific, 2004.

\bibitem{Haimo1986:SIAM_JCO}
V.T. Haimo.
\newblock Finite time controllers.
\newblock {\em SIAM Journal of Control and Optimization}, 24(4):760--770, 1986.

\bibitem{Holloway2018:PhD}
J.C. Holloway.
\newblock {\em Prescribed Time Stabilization and Estimation for Linear Systems
  with Applications in Tactical Missile Guidance}.
\newblock PhD thesis, University of California, San Diego, 2018.

\bibitem{Hong2001:Aut}
Y.~Hong.
\newblock H$_{\infty}$ control, stabilization, and input-output stability of
  nonlinear systems with homogeneous properties.
\newblock {\em Automatica}, 37(7):819--829, 2001.

\bibitem{KarafyllisKrstic2017:Book}
I.~Karafyllis and M.~Krstic.
\newblock {\em Predictor Feedback for Delay Systems: Implementations and
  Approximations}.
\newblock Birkhauser, 2017.

\bibitem{KarafyllisKrstic2018:Book}
I.~Karafyllis and M.~Krstic.
\newblock {\em Input-to-State Stability for PDEs}.
\newblock Springer, 2018.

\bibitem{Kawski1991:ACDS}
M.~Kawski.
\newblock Families of dilations and asymptotic stability.
\newblock {\em Analysis of Controlled Dynamical Systems}, pages 285--294, 1991.

\bibitem{Korobov1979:DAN}
V.I. Korobov.
\newblock A solution of the synthesis problem using controlability function.
\newblock {\em Doklady Academii Nauk SSSR}, 248:1051--1063, 1979.

\bibitem{Krstic2008:Aut}
M.~Krstic.
\newblock Lyapunov tools for predictor feedbacks for delay systems: Inverse
  optimality and robustness to delay mismatch.
\newblock {\em Automatica}, 44:2930--2935, 2008.

\bibitem{Krstic2009:Book}
M.~Krstic.
\newblock {\em Delay Compensation for Nonlinear, Adaptive and PDE systems}.
\newblock Birkhauser, 2009.

\bibitem{LaSalle1958:PNAS}
J.~La~Salle.
\newblock Time optimal control systems.
\newblock {\em Proceedings of the National Academy of Sciences of the United
  States of America}, 45(4):573--577, 1958.

\bibitem{Levant2005:Aut}
A.~Levant.
\newblock Homogeneity approach to high-order sliding mode design.
\newblock {\em Automatica}, 41(5):823--830, 2005.

\bibitem{Majda1975:IUMJ}
A.~Majda.
\newblock Disappearing solutions for the dissipative wave equation.
\newblock {\em Indiana University Mathematics Journal}, 24(12):1119--1133,
  1975.

\bibitem{ManitiusOlbrot1979:TAC}
A~Manitius and A.W. Olbrot.
\newblock Finite spectrum assignment problem for systems with delays.
\newblock {\em IEEE Transactions on Automatic Control}, 24(4):541--553, 1979.

\bibitem{MironchenkoPrieur2020:SIAM}
A.~Mironchenko and C.~Prieur.
\newblock Input-to-state stability of infinite dimensional systems: Recent
  results and open questions.
\newblock {\em SIAM Review}, 62(3):529--614., 2020.

\bibitem{Nakamura_etal2002:SICE}
H.~Nakamura, Y.~Yamashita, and H.~Nishitani.
\newblock Smooth {Lyapunov} functions for homogeneous differential inclusions.
\newblock In {\em Proceedings of the 41st SICE Annual Conference}, pages
  1974--1979, 2002.

\bibitem{Orlov2005:SIAM_JCO}
Y.~Orlov.
\newblock Finite time stability and robust control synthesis of uncertain
  switched systems.
\newblock {\em SIAM Journal of Control and Optimization}, 43(4):1253--1271,
  2005.

\bibitem{Orlov2020:Book}
Y.~Orlov.
\newblock {\em Nonsmooth Lyapunov Analysis in Finite and Infinite Dimensions}.
\newblock Springer, 2020.

\bibitem{Orlov2022:Aut}
Y.~Orlov.
\newblock Time space deformation approach to prescribed-time stabilization:
  Synergy of time-varying and non-lipschitz feedback designs.
\newblock {\em Automatica}, 144(10):110485, 2022.

\bibitem{Pazy1983:Book}
A.~Pazy.
\newblock {\em Semigroups of Linear Operators and Applications to Partial
  Differential Equations}.
\newblock Springer, 1983.

\bibitem{Perruquetti_etal2008:TAC}
W.~Perruquetti, T.~Floquet, and E.~Moulay.
\newblock Finite-time observers: application to secure communication.
\newblock {\em IEEE Transactions on Automatic Control}, 53(1):356--360, 2008.

\bibitem{Polyakov2012:TAC}
A.~Polyakov.
\newblock Nonlinear feedback design for fixed-time stabilization of linear
  control systems.
\newblock {\em IEEE Transactions on Automatic Control}, 57(8):2106--2110, 2012.

\bibitem{Polyakov2020:Book}
A.~Polyakov.
\newblock {\em Generalized Homogeneity in Systems and Control}.
\newblock Springer, 2020.

\bibitem{Polyakov_etal2023:Aut}
A.~Polyakov, D.~Efimov, and X.~Ping.
\newblock Consistent discretization of homogeneous finite/fixed-time
  controllers for {LTI} systems.
\newblock {\em Automatica}, 2023.

\bibitem{PolyakovKrstic2023:TAC}
A.~Polyakov and M.~Krstic.
\newblock Finite-and fixed-time nonovershooting stabilizers and safety filters
  by homogeneous feedback.
\newblock {\em IEEE Transaction on Automatic Control}, 2023.

\bibitem{Rosier1992:SCL}
L.~Rosier.
\newblock Homogeneous {L}yapunov function for homogeneous continuous vector
  field.
\newblock {\em Systems \& Control Letters}, 19:467--473, 1992.

\bibitem{Ryan1995:SCL}
E.P. Ryan.
\newblock Universal stabilization of a class of nonlinear systems with
  homogeneous vector fields.
\newblock {\em Systems \& Control Letters}, 26:177--184, 1995.

\bibitem{Shinar_etal2014:JFI}
J.~Shinar, V.Y. Glizer, and V.~Turetsky.
\newblock Capture zone of linear strategies in interception problems with
  variable structure dynamics.
\newblock {\em Journal of The Franklin Institute}, 351(4):2378--2395, 2014.

\bibitem{Song_etal2017:Aut}
Y.~Song, Y.~Wang, J.~Holloway, and M.~Krstic.
\newblock Time-varying feedback for regulation of normal-form nonlinear systems
  in prescribed finite time.
\newblock {\em Automatica}, 83:243--251, 2017.

\bibitem{Sontag1989:TAC}
E.D. Sontag.
\newblock Smooth stabilization implies coprime factorization.
\newblock {\em IEEE Transactions on Automatic Control}, 34:435--443, 1989.

\bibitem{Steeves_etal2020:EJC}
D.~Steeves, M.~Krstic, and R.~Vazquez.
\newblock Prescribed-time estimation and output regulation of the linearized
  schr\"odinger equation by backstepping.
\newblock {\em European Journal of Control}, 55(9):3--13, 2020.

\bibitem{Utkin1992:Book}
V.~I. Utkin.
\newblock {\em Sliding Modes in Control Optimization}.
\newblock Springer-Verlag, Berlin, 1992.

\bibitem{Zekraoui_etal2023:Aut}
S.~Zekraoui, N.~Espitia, and W.~Perruquetti.
\newblock Finite/fixed-time stabilization of a chain of integrators with input
  delay via pde-based nonlinear backstepping approach.
\newblock {\em Automatica}, 155:111095, 2023.

\bibitem{Zimenko_etal2020:TAC}
K.~Zimenko, A.~Polyakov, D.~Efimov, and W.~Perruquetti.
\newblock Robust feedback stabilization of linear mimo systems using
  generalized homogenization.
\newblock {\em IEEE Transactions on Automatic Control}, 2020.

\bibitem{Zubov1964:Book}
V.~I. Zubov.
\newblock {\em Methods of A.M. Lyapunov and Their Applications}.
\newblock Noordhoff, Leiden, 1964.

\bibitem{Zubov1958:IVM}
V.I. Zubov.
\newblock On systems of ordinary differential equations with generalized
  homogeneous right-hand sides.
\newblock {\em Izvestia vuzov. Mathematica (in Russian)}, 1:80--88, 1958.

\end{thebibliography}

\end{document}